\documentclass[12pt,a4paper]{article}
\usepackage[utf8]{inputenc}
\usepackage[english]{babel}
\usepackage{amsmath}
\usepackage{amsfonts}
\usepackage{amssymb}
\usepackage{amsthm}
\usepackage{bm}
\usepackage{float}
\usepackage{graphicx}
\usepackage{bbm}
\usepackage{hyperref}
\usepackage{subcaption}
\usepackage[]{appendix}
\usepackage{caption}
\usepackage{setspace} 
\usepackage{geometry}
\usepackage{pdfpages}
\usepackage{algorithm}
\usepackage{algpseudocode}
\usepackage{hyperref}
\newtheorem{theorem}{Theorem}[section]
\newtheorem{lemma}[theorem]{Lemma}

\newtheorem{remark}[theorem]{Remark}
\newcommand{\easysum}[2]{\ensuremath{\underset{#1}{\overset{#2}{\sum}}}}
\newcount\colveccount
\newcommand*\colvec[1]{
        \global\colveccount#1
        \begin{pmatrix}
        \colvecnext
}

\newcommand{\R}{\mathbb R}
\newcommand{\diam}{\mathrm {diam}}
\def\colvecnext#1{
        #1
        \global\advance\colveccount-1
        \ifnum\colveccount>0
                \\
                \expandafter\colvecnext
        \else
                \end{pmatrix}
        \fi
}
\graphicspath{{./plots/}}

\DeclareMathOperator*{\argmin}{arg\,min}

\begin{document}

 \author{Florian Heinemann\thanks{Institute for Mathematical Stochastics, University of G\"ottingen, Goldschmidtstra{\ss}e 7, 37077 G\"ottingen} \and Axel Munk \footnotemark[1] \thanks{Max Planck Institute for Biophysical
		Chemistry, Am Fa{\ss}berg 11, 37077 G\"ottingen}\and Yoav Zemel \thanks{Centre for Mathematical Sciences, University of Cambridge, Cambridge CB3 0WB} }
\title{Randomised Wasserstein Barycenter Computation: Resampling with Statistical Guarantees}
\date\today
\maketitle

\begin{abstract}
We propose a hybrid resampling method to approximate finitely supported Wasserstein barycenters on large-scale datasets, which can be combined with any exact solver. Nonasymptotic bounds on the expected error of the objective value as well as the barycenters themselves allow to calibrate computational cost and statistical accuracy.  The rate of these upper bounds is shown to be optimal and independent of the underlying dimension, which appears only in the constants.  Using a simple modification of the subgradient descent algorithm of Cuturi and Doucet, we showcase the applicability of our method on a myriad of simulated datasets, as well as a real-data example from cell microscopy which are out of reach for state of the art algorithms for computing Wasserstein barycenters.
\end{abstract}

\section{Introduction}
Recently, optimal transport (OT), and more specifically the Kantorovich (also known as Wasserstein) distance, have achieved renewed interested as they have been recognised as attractive tools in data analysis. Despite its conceptual appeal in many applications (e.g., Rubner et al.\ \cite{rubner2000earth}; Evans and Matsen \cite{evans2012phylogenetic}; Klatt et al.\ \cite{klatt2020empirical}), optimal transport-based data analysis has been triggered on the one hand by recent computational progress (see e.g., Peyr\'e and Cuturi \cite{peyre2019computational}, Schmitzer \cite{schmitzer2019stabilized}, Solomon et al.\ \cite{solomon2015convolutional}, Altschuler et al.\ \cite{altschuler2017near}) and on the other hand by a refined understanding of its statistical properties when estimated from data (e.g., del Barrio et al.\ \cite{del1999central}; Sommerfeld et al.\ \cite{sommerfeld2018inference}; Weed and Bach \cite{weed2019sharp}). This also lead to an increasing interest in Fr\'echet means, or barycenters, with respect to that distance. Since their introduction in the landmark paper of Agueh and Carlier \cite{aguehBarycentersWassersteinSpace2011}, interest in the so-called Wasserstein barycenters sparked.  Among the plethora of their potential applications, one can name unsupervised dictionary learning (Schmitz et al.\ \cite{schmitz2018wasserstein}), distributional clustering (Ye et al.\ \cite{ye2017fast}), Wasserstein principal component analysis (Seguy and Cuturi \cite{seguy2015principal}), neuroimaging (Gramfort et al.\ \cite{gramfort2015fast}) and computer vision (Rabin et al.\ \cite{rabin2011wasserstein}; Solomon et al.\ \cite{solomon2015convolutional}; Bonneel et al.\ \cite{bonneel2016wasserstein}).  An appealing feature of Wasserstein barycenters is their ability to successfully capture geometric properties of complex data objects, thus allowing to define a meaningful notion of average for such objects.  

Figure \ref{FigBaryEllipses} illustrates this by displaying the, numerically computed, Wasserstein barycenter of a collection of ellipses (associated with uniform probability measures on their surface) along with their pixel-wise mean.  It is seen that the former is a single ``average" ellipse, whilst the latter gives a blurry object that is not representative of a ``typical realisation" from this dataset. In fact, in this case the collection of probability measures forms a location-scatter family, and it is known (\'Alvarez et al.\ \cite{alvarez2016fixed}) that the barycenter must again lie in the same family.  Hence, the OT-barycenter of these ellipses is again an ellipse.
\begin{figure}
\includegraphics[width=\textwidth]{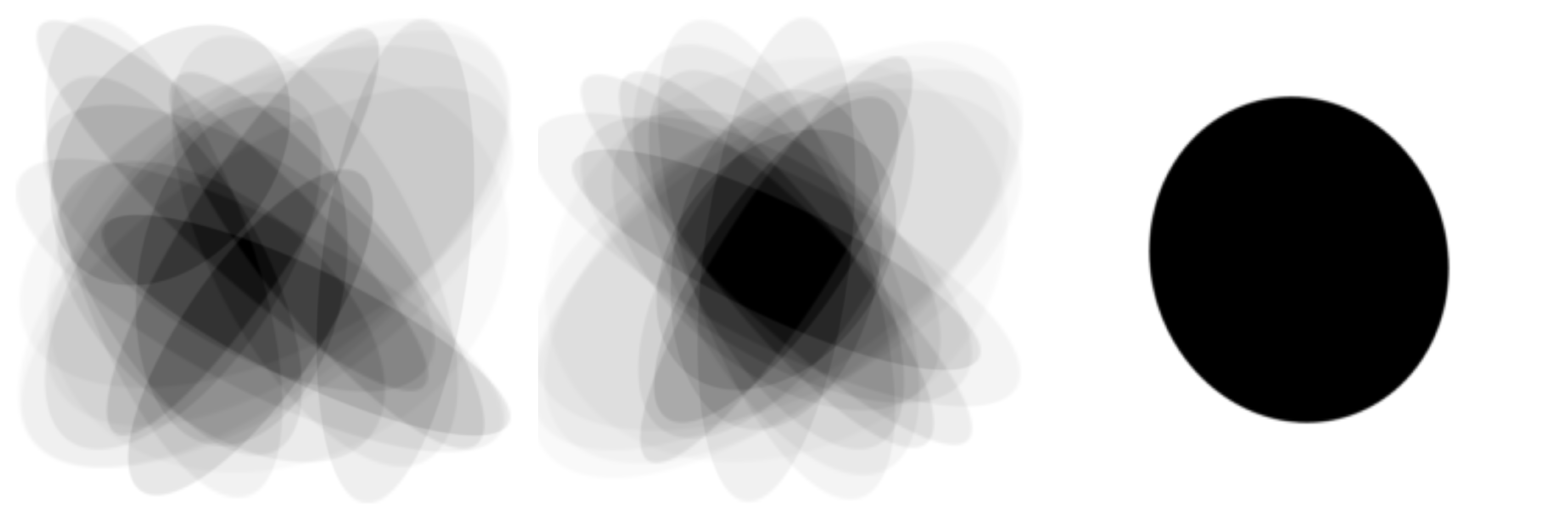}
\caption{From left to right: Mean, Mean after recentering and Wasserstein barycenter of 20 randomly generated ellipses in $\mathbb R^2$.}
\label{FigBaryEllipses}
\end{figure}
\subsection{The OT-Barycenter Problem}\label{sec:foundations}
One of the most fundamental questions in statistics and data analysis is inferring the mean of a random quantity on the basis of realisations $x_1,\dots,x_N$ thereof.  Whilst this is conceptually straightforward when the data lie in a Euclidean space or a Hilbert space, many datasets exhibit complex geometries that are far from Euclidean (e.g., Billera et al.\ \cite{billera2001geometry}; Dryden et al.\ \cite{dryden2009non}; Bronstein et al.\ \cite{bronstein2017geometric}), catalysing the emergence of the field of non-Euclidean statistics (e.g., Patrangenaru and Ellingson \cite{patrangenaru2015nonparametric}; Huckemann and Eltzner \cite{huckemann2021data}; Dryden and Marron \cite{dryden2021object}).  Utilising the fact that the mean of points $x_1,\dots,x_N\in\mathbb R^D$ is characterised as the unique minimiser of $x\mapsto \sum_{i=1}^N\|x_i-x\|^2$, the notion of a barycenter (or Fr\'echet mean; Fr\'echet \cite{frechet1948elements}; Huckemann et al.\ \cite{huckemann2010intrinsic}) extends this to the non-Euclidean case by replacing the norm $\|x_i-x\|$ with an arbitrary distance function.  Motivated from the preceding paragraphs, our work focusses on OT-barycenters, where the distance defining the barycenter is an optimal transport distance.  More specifically, in this paper we are concerned with OT-barycenters of finitely supported probability measures i.e., collections of weighted points. To set up notation, consider a finite metric space $(\mathcal{X},d)$ and $N$ measures of the form
\begin{align*}
\mu_i=\easysum{k=1}{M_i} b_k^i \delta_{x_k^i},
\quad x_k^i\in \mathcal X,\quad b_k^i\ge0,
\quad \sum_{k=1}^{M_i}b_k^i=1,
\qquad i=1,\dots,N,
\end{align*}
where $\delta_x$ is a Dirac measure at $x\in\mathcal X$. In imaging applications, the $x_k^i=x^i$'s could be points on a regular grid in $[0,1]^2$ with the weights $b_k^i$ being the greyscale intensity of image $k$ at pixel $x^i$.  The above, more general formulation allows to also treat irregular points clouds that are unrelated to each other. In this setting the $p$-Wasserstein distance (Kantorovich \cite{kantorovich1942translocation}; Vaserstein \cite{vaserstein1969markov}) between any two of these measures, $\mu_i$ and $\mu_j$, is
\begin{align*}
W_p(\mu_i,\mu_j)=
\bigg(\underset{\pi \in \Pi (\mu_i , \mu_j)}{\min} \quad &\easysum{k=1}{M_i}\easysum{l=1}{M_j}\pi_{kl}d(x^i_k,x^j_l)^p
\bigg)^{1/p},
\end{align*}
where $p\ge1$ and the set of \textbf{couplings} between $\mu_i$ and $\mu_j$ contains all the joint distributions on $\mathcal X^2$ having $\mu_i$ and $\mu_j$ as marginal distributions, and is given by
\begin{align*}
\Pi (\mu_i , \mu_j):=
\left\{ \pi \in \mathbb{R}^{M_i \times M_j} \vert \pi \mathbf{1}_{M_i} = b_i ,  \mathbf{1}_{M_j}^T \pi=b_j^T\right\}
,\qquad \mathbf{1}_M = (1,1,\dots,1)\in \R^M.
\end{align*}
In particular, it can be shown that $W_p$ is a distance on the space of probability measures on $\mathcal{X}$ (see e.g., Villani \cite[Chapter 6]{villaniOptimalTransportOld2009}; or Zolotarev \cite{zolotarev1976metric} for an early reference). To define barycenters with respect to this distance we need to choose an ambient space $\mathcal{Y}$ containing $\mathcal{X}$, since any reasonable choice of such a barycenter should be allowed to have mass at positions that may differ from the support points of the $N$ data measures. Let $\mathcal P(\mathcal Y)$ be the space of measures on $\mathcal{Y}$.  The \textbf{$p$-Fr\'echet functional} $F: \mathcal{P}(\mathcal{Y}) \rightarrow \mathbb{R}_+$ associated with $\mu_1,\dots ,\mu_N$ is defined as
\begin{equation}\label{def::frechet}
F^p(\mu)
=\frac{1}{N}\easysum{i=1}{N} W_p^p(\mu_i,\mu).
\end{equation}
\begin{figure}[H]
\begin{center}
  \includegraphics[scale=0.06]{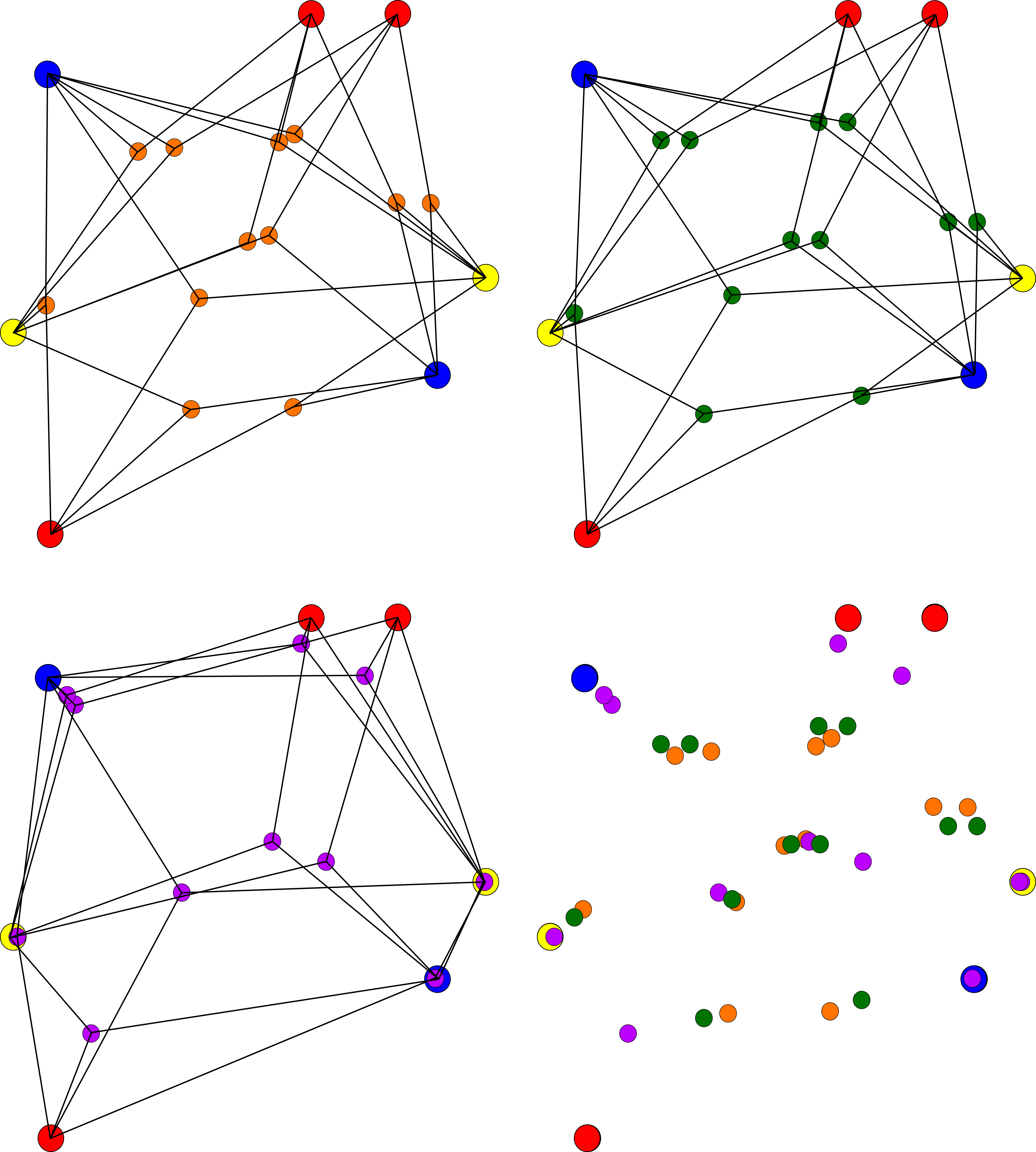}
  \caption{Three different centroid sets for the same set of three measures (red, blue, yellow). The lines connect to the support points of the measures from which a centroid point has been constructed. Here, it holds $\lvert \mathcal{C} \rvert=12$. \textbf{Top left:} Centroids for $p=3$ in orange. \textbf{Top right:} Centroids for $p=2$ in green. \textbf{Bottom left:} Centroids for $p=1$ in violet. \textbf{Bottom right:} All three centroids superimposed onto each other.}
   \label{centroids}
\end{center}
\end{figure}
We call any minimiser of $F^p$ a \textbf{$p$-Fr\'echet mean} or \textbf{$p$-Wasserstein barycenter} of $\mu_1,\dots ,\mu_N$. When $p$ is omitted, it is assumed to be equal to 2.

The starting point of our paper is an observation made by Le Gouic and Loubes \cite[Theorem 8]{le2017existence}, who showed that when the ambient $\mathcal{Y}$ is geodesic (e.g., $\mathcal Y=\R^D$), any $p$-barycenter of $\mu_1,\dots ,\mu_N$ is supported on the $p$-centroid set
\begin{align*}
\mathcal{C}:=\left\{\underset{y \in \mathcal Y }{\argmin} \easysum{i=1}{N}d^p(x_i,y)\Big\vert x_i \in \text{supp}(\mu_i)\right\}
=\left\{\underset{y \in \mathcal Y }{\argmin} \easysum{i=1}{N}d^p(x^i_{k_i},y)\Big\vert k_i\in\{1,\dots,M_i\} \right\}
,
\end{align*}
where, by a slight abuse of notation, $d(\cdot,\cdot)$ denotes a metric on $\mathcal Y$ that extends the original metric $d$ on $\mathcal X$ and supp denotes the support of the measure. Thus, even though $\mathcal Y$ is usually not a finite space, the minimisation can always be carried out on measures supported on a finite set $\mathcal C\subseteq \mathcal Y$.  For this reason, defining $W_p$ for finitely supported measures suffices for the purpose of the present paper. In this setting of finitely supported measures barycenters always exist, but they are not necessarily unique.

The choice of the ambient space $\mathcal Y$ is not unique.  Whilst it can typically assumed to be a Euclidean space $\mathbb R^D$, certain applications might warrant setting it to be a curved space of lower dimension $D'<D$ to better represent the given data. This could for instance apply to data on a sphere, where it might be natural to restrict the barycenter to this set as well.  The upper bounds we provide in Section~\ref{sec:bounds} depend on the dimension of the chosen ambient space $\mathcal Y$.  It is usually the case that our bounds improve with decreasing dimension of the ambient space. Therefore, choosing a curved low-dimensional ambient space may not only better capture the structure of the data, but also improve the statistical guarantees on the approximation error.

We would like to stress that the definition and our later algorithms can be adapted in a straightforward way to Fr\'echet means with non-uniform weights other than $1/N$ in \eqref{def::frechet}. We avoid this generality for brevity and simplicity.
When $p=2$ and $d$ is the Euclidean metric, the argmin is the linear average $\sum x_i/N$ and $\mathcal{C}$ simplifies to
\begin{align*}
\mathcal{C}:=\left\{ \frac{1}{N}\easysum{i=1}{N}x_i\Big\vert x_i \in \text{supp}(\mu_i)\right\}
=\left\{ \frac1N\easysum{i=1}{N}x^i_{k_i}\Big\vert k_i\in\{1,\dots,M_i\} \right\}
.
\end{align*}
As the set $\mathcal C$ plays an important role in the following, it is illustrated in Figure~\ref{centroids} for different powers of the Euclidean metric in $\mathbb R^2$.
In particular, $\mathcal{C}$ is finite with cardinality at most $\prod_{i=1}^NM_i$ and the Fr\'echet functional can be expressed as a sum of sums over at most $\lvert \mathcal{C}\rvert \sum_{i=1}^{N}M_i$ summands. Identifying any candidate barycenter $\mu$ with its weights vector $a \in \mathbb{R}^{\lvert \mathcal{C}\rvert}$ on $\mathcal{C}$, this allows to rewrite the minimisation of the $p$-Fr\'echet functional as a linear program (LP) which computes optimal transport plans $\pi^{(i)}\in \Pi (\mu,\mu_i)$ between $\mu$ and $\mu_i$ for $i=1, \dots ,N$, whilst at the same time also minimising over the weights vector $a$ of $\mu$. The LP can then be denoted as

\begin{equation}
\label{eq:linProgBary}
\begin{aligned}
\underset{\pi^{(1)},\dots,\pi^{(N)},a}{\min} \quad & \frac{1}{N} \easysum{i=1}{N}\easysum{j=1}{\lvert \mathcal{C} \rvert}\easysum{k=1}{M_i}\pi^{(i)}_{jk}c_{jk}^i \\
\text{subject to} \quad &\easysum{k=1}{M_i}\pi^{(i)}_{jk}=a_j \quad \forall \ i=1,\dots ,N, \quad \forall j=1,\dots,\lvert \mathcal{C} \rvert \\
&\easysum{j=1}{\lvert \mathcal{C} \rvert}\pi^{(i)}_{jk}=b^i_k \quad \forall \ i=1,\dots ,N ,\quad \forall k=1,\dots ,M_i \\
&\pi^{(i)}_{jk}\geq 0 \quad \forall i=1,\dots ,N \quad \forall j=1,\dots ,\lvert \mathcal{C}\rvert ,\quad \forall k=1,\dots ,M_i,
\end{aligned}
\end{equation}
where $c_{jk}^i=d^p(\mathcal C_j,x_{k}^i)$ is the $p$-th power of the distance between the $j$-th point of $\mathcal{C}$ and the $k$-th point in the support of $\mu_i$.

Unfortunately, the size of $\mathcal{C}$ typically grows as $\prod M_i$, rendering this linear program intractable for even moderate values of $N$ and $M_i$ (see Figure~\ref{centroids}).  Borgwardt and Patterson \cite{borgwardt2020improved} provide some reformulations and improvements, but ultimately, this approach is currently infeasible for meaningful applications to large-scale data. To illustrate the size of this linear program consider a set of $N=100$ greyscale images of size $M_i=M=256\times256$. We assume the images to be probability measures supported on an equidistant grid in $[0,1]^2$ equipped with the Euclidean distance. Even if we exploit the fact that for $p=2$ the set $\mathcal{C}$ is simply an $N$-times finer grid (see Anderes et al.\ \cite{anderesDiscreteWassersteinBarycenters2016}), this LP still has over $10^{15}$ variables and over $10^{10}$ constraints. If the support points would be in more general positions the LP would have over $10^{488}$ variables and over $10^{483}$ constraints. To put this into perspective, we note that in November 2020 the highest ranking supercomputer in the TOP500 list (for details on this list see Dongarra et al.\ \cite{dongarra2003linpack}) had about $10^{15}$ bytes of RAM. 

\subsection{OT-Barycenter Computation}
To overcome this computational obstacle and to utilise the impressive geometrical power of the OT-barycenter (recall Figure~\ref{FigBaryEllipses}), there has been a great effort in constructing algorithms that yield approximations of the OT-barycenter whilst reducing the computational complexity of the exact LP-formulation by several orders of magnitude. In fact, it is known that already for three measures in two dimensions finding an exact Wasserstein barycenter is NP-hard, in general (Borgwardt and Patterson \cite{borgwardt2019computational}).

Cuturi and Doucet \cite{cuturiFastComputationWasserstein2014} proposed a subgradient descent method to compute the best approximation of the Wasserstein barycenter, which is supported on a prespecified support set. If we choose this set to be equal to $\mathcal{C}$, then this algorithm approximates a true barycenter of the measures. However, as we discussed previously, the set $\mathcal{C}$ is far too large for efficient computations (see Figure~\ref{centroids} for a small scale example). Its size, which controls the number of variables in the linear program \eqref{eq:linProgBary}, can be as large as the support size of the product measure $\mu_1 \otimes \dots \otimes \mu_N$ of the $N$ data measures.  Cuturi and Doucet also introduced an alternating, Lloyd-type procedure, that switches between optimising the weights of the barycenter measure on a specific support set and updating the support set according to the new weights. Since there always exists a barycenter with support size bounded above by $\sum_{i=1}^{N}\lvert M_i\rvert -N+1$ (Anderes et al.\ \cite{anderesDiscreteWassersteinBarycenters2016}), we can simply choose a support set of this size to approximate the true barycenter. This option to compute an exact barycenter without using the set $\mathcal{C}$ comes at a tangible cost, however. This alternating procedure is highly runtime extensive, since it involves solving the fixed-support barycenter problem in each step. Still, it is vastly superior to the direct LP-approach. However, this alternating procedure yields a non-convex minimisation problem, which is prone to converge to local minima instead of global ones.

The computational burden can be alleviated by means of regularisation, most commonly in the form of an additive entropy penalty, which leads to the so-called Sinkhorn distance (Cuturi \cite{cuturi2013sinkhorn}).  Cuturi and Doucet \cite{cuturiFastComputationWasserstein2014} compute the barycenter with respect to  the Sinkhorn distance by means of subgradient descent.   Benamou et al.\ \cite{benamou2015iterative} exploit the relation of entropy to the Kullback--Leibler divergence and solve this regularised problem by iterative Bregman projections. These approaches reduce the runtime by orders of magnitude, and approximate a regularised surrogate for the exact Wasserstein barycenter. Unfortunately, (entropic) regularisation is not without its drawbacks, particularly the choice of the regularisation parameter is a delicate issue. Although the regularised solution converges towards the real solution with maximal entropy as the regularisation parameter vanishes (\cite{benamou2015iterative}), computations become costly and numerically unstable for small parameters (Altschuler et al.\ \cite{altschuler2017near}; Dvurechensky et al.\ \cite{dvurechensky2018computational}; Klatt et al.\ \cite{klatt2020empirical}), which reflects a "no free lunch" scenario. Workarounds like log-stabilisations (Schmitzer \cite{schmitzer2019stabilized}) exist, but do not completely address this problem, as with increasingly smaller regularisation there can still be numerical instabilities, and the stabilisation sacrifices a significant portion of the computational gain resulting from the regularisation.  Moreover, whilst runtime and memory requirements of the regularised methods scale linearly in the number of measures, the scaling in the support sizes of these measures is still problematic.  There have been multiple iterations of improvements on the naive Sinkhorn algorithm (e.g.,  \cite{altschuler2017near,dvurechensky2018computational}), but the approximation of Wasserstein barycenters already for relatively small ensembles of medium sized data objects, say $100$ images with around $10^5$ pixels each, becomes nevertheless infeasible in terms of computation time and required memory.

Recently, Xie et al.\ \cite{xie2020fast} have proposed an inexact proximal point method which allows to solve the fixed-support OT-barycenter problem at a similar complexity as the regularised problem, whilst avoiding to introduce a regularisation term. Their simulations also suggest that they achieve sharper images as barycenters compared to the standard entropy-regularised methods. However, their method has an increased memory demand and is still not applicable to large scale data as the scaling in the support size is still identical to the Sinkhorn algorithm.  For further recent contributions, see e.g., Tiapkin et al.\ \cite{tiapkin2020stochastic}; Ge et al.\ \cite{ge2019interior}; Dvurechenskii et al.\ \cite{dvurechenskii2018decentralize}; Lin et al.\ \cite{lin2020computational}; Li et al.\ \cite{li2020continuous}.

To summarise the above, the computational aspect of the OT-barycenter problem has attracted significant interest over the last years and an ever growing and improving toolbox of methods is available to tackle this problem. However, the size of problems which can be solved is still rather limited; modern applications e.g.\ in medical high-resolution imaging, are currently out of reach. Additionally, most methods consider the fixed-support OT-barycenter problem, where all measures are assumed to be supported on the same finite set. These methods cannot exploit sparsity in the support of the measures if it is present. In particular, if the measures have vastly different supports, fixed support methods are likely to be inefficient.

\subsection{Our Approach}
In this work we are primarily concerned with measures having potentially very different supports (as in the synthetic example in Figure~\ref{centroids}). We decrease the problem size by generating random sparse approximations to the data measures.  This will have a very limited effect on the runtime of fixed-support methods, but, for instance, the original, alternating subgradient descent by Cuturi and Doucet can fully exploit the advantages of this reduced problem size (see Section~\ref{sec:computation}). In the same spirit any other method which takes advantage of sparse support aligns well with our resampling method. We stress that the proposed resampling method to compute random approximations of Wasserstein barycenters does not rely on any type of regularisation.

Our approach extends work by Sommerfeld et al.\ \cite{sommerfeld2019optimal} on randomised optimal transport computation to the barycenter problem by replacing the original measures by their empirical counterparts, obtained from $S$ independent random samples. Then, the true barycenter $\mu^*$ is estimated by the barycenter of the empirical measures $\widehat \mu_S^*$. We provide nonasymptotic $L_1$-type bounds on the objective values of $\mu^*$ and $\widehat\mu^*_S$, as well as on the distance between the set of empirical barycenters and the true ones.  These bounds are optimal in terms of the dependence on the resampling size $S$. As an example, we will see that for $N$ measures with $M$ support points in $[0,1]^D$, it holds that
\begin{align}\label{eq:p2}
\mathbb{E}[\lvert F^2(\mu^*)-F^2(\widehat \mu_S^*) \rvert ] \leq 4D^{3/2}S^{-\frac12}\begin{cases}
2+\sqrt 2 & D=1\\
2+\log_2(M) & D=2\\
(3+\sqrt 2)M^{\frac12-\frac1D}  & D\ge3 
\end{cases}
\end{align}
where $F^2$ corresponds to $p=2$ in \eqref{def::frechet}, i.e., the two barycenters are compared with respect to the $W_2^2$ distance.  Two key aspects of our upper bound appear striking to us. First, the bound is uniform in the number of measures $N$, which has important consequences for randomised computations on a large number of datasets; and second, the convergence rate is $S^{-\frac{1}{2}}$, independently of the dimension $D$ of the ambient space $\mathcal{Y}$ in which the measures reside (see \cite{sommerfeld2019optimal} for $N=2$). This is in contrast to the case of absolutely continuous measures, where this rate is typically achieved when $D<2p$, whereas if $D>2p$, convergence holds at the slower rate $S^{-\frac{p}{D}}$ (see e.g., Dudley \cite{dudley1969speed};  Dereich et al.\ \cite{dereich2013constructive}; Lei \cite{lei2020convergence} and references therein), exhibiting a curse of dimensionality. In Figure~\ref{3d_bary} we showcase two examples of sampling approximations of barycenters in three dimensions, where we obtain visually striking approximations of a barycenter whilst reducing the support size of the measures by an order of magnitude.  The bound \eqref{eq:p2} is particular case of our general results as stated in Theorem~\ref{full_frechetbound} below.  It is important to stress that the upper bound depends on the intrinsic dimension of $\mathcal X$ and not on that of the ambient space $\mathcal Y$.  If, for example, $\mathcal X$ can be embedded in a set that is a Lipschitz image of a lower dimensional cube $[0,1]^{D'}$ with $D'<D$, then \eqref{eq:p2} still holds with $D$ replaced by $D'$. This is analogous to a similar well-known phenomenon in manifold learning (e.g., Genovese et al.\ \cite{genovese2012minimax}).

\begin{figure}
    \centering
    \includegraphics[width=0.85\textwidth]{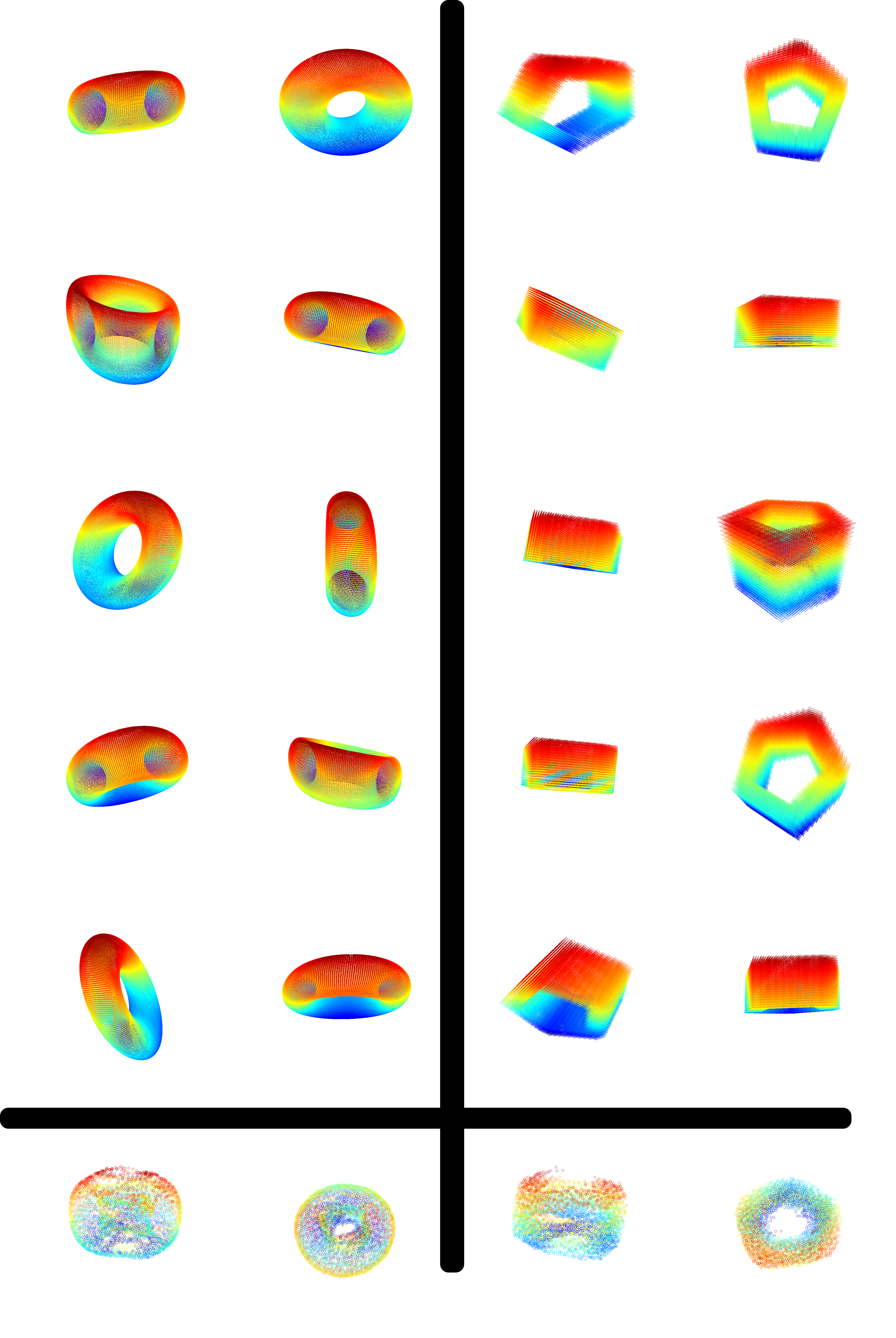}
       \caption{\textbf{Upper Left:} Projections of $N=10$ torsos in $\mathbb{R}^3$, discretised on about $M=62500$ points. \textbf{Lower Left:} Projected stochastic barycenter approximation with $S=4000$ from the same angle (left) and from an adjusted angle (right). \textbf{Upper Right:} Projections of $N=10$ pentagonal prisms in $\mathbb{R}^3$, discretised on about $M=85000$ points. \textbf{Lower Right:} Projected stochastic barycenter approximation with $S=4000$ from the same angle (left) and from an adjusted angle (right).}
     \label{3d_bary}
\end{figure}
To the best of our knowledge, the best complexity bounds on the computation of the optimal transport for two measures are of order $\tilde{\mathcal{O}}(M^{\frac{5}{2}})$ (Lee and Sidford \cite{lee2014path}), i.e., this resampling scheme reduces the runtime by a factor of order $\tilde{\mathcal{O}}(\left(\frac{M}{S}\right)^{\frac{5}{2}})$.  For instance, by resampling $15\%$ of the data points, the runtime is reduced by a factor of more than 100. For $N=2$ Sommerfeld et al.\ \cite{sommerfeld2019optimal} report empirical relative errors of around $5\%$ for $S=2000$ on $128\times 128$ images, which translates to $S\approx 0.12M$. These results are striking and provided the motivation to extend this method to the barycenter problem for general $N$, which suffers even more than the optimal transport problem from prohibitive computational cost.

Furthermore, in the case $D=2$, the upper bound depends only logarithmically on the support size $M$. Thus, in the (ever important) case of two dimensional images, we have only logarithmic dependence on the resolution of the image. This suggests a good performance of our sampling method on high-resolution image-datasets. We explore this empirically in Section~\ref{sec:computation}.

Finally, to perform our simulations, we leverage some empirical observations on the barycenters of uniform, finitely supported measures to modify an existing iterative method for our needs in order to obtain a stochastic method that excels at solving the barycenter on large, sparsely supported datasets, which we in the following refer to as \emph{Stochastic-Uniform-Approximation (SUA)--method}. We explore the quality of our bounds empirically in large-scale simulation studies and test the visual performance of our method on artificial and real data.

\bigskip

\textbf{Population Barycenters}.
Let us stress that in this paper, the collection of measures $\mu_1,\dots,\mu_N$ is viewed as fixed.  A different but related problem is that of estimating a population barycenter when the measures $\mu_i$ are realisations of a random probability measure $\bm\mu$ in $\mathcal P(\mathcal Y)$ (i.e., the distribution of $\bm\mu$ is an element of $\mathcal P(\mathcal P(\mathcal Y)))$; see e.g., Pass \cite{pass2013optimal}; Bigot and Klein \cite{bigot2018characterization}.  The minimiser of \eqref{def::frechet} is then an empirical barycenter, whose asymptotic properties (as $N\to\infty$) are studied in the form of consistency (Le Gouic and Loubes \cite{le2017existence}), rates of convergence (Ahidar-Coutrix et al.\ \cite{ahidar2020convergence}; Le Gouic et al.\ \cite{gouic2019fast}) and, in very specific cases, central limit theorems (Panaretos and Zemel \cite{panaretos2016amplitude}; Agueh and Carlier \cite{agueh2017vers}; Kroshnin et al.\ \cite{kroshnin2019statistical}). From a computational perspective, (stochastic) gradient descent-type methods have been shown to converge in some situations (Zemel and Panaretos \cite{zemel2019frechet}; Backhoff-Varaguas et al.\ \cite{backhoff2018bayesian}; Chewi et al.\ \cite{chewi2020gradient}).  Most of these papers focus on absolutely continuous measures that are fully observed; closer to our context is the recent work of Dvinskikh \cite{dvinskikh2020sa}, where the measures are discrete and the regularised and unregularised settings are both discussed.  In summary, the stochasticity in the two problems is of a different nature: In the population barycenter context, it arises in $\mathcal P(\mathcal Y)$, i.e., at the level of the measures; whereas in our setting it takes place on the space $\mathcal Y$ itself, i.e., at the level of each individual measure $\mu_i$ via the resampling.

\bigskip

\textbf{Deterministic Approaches.}  Instead of generating \emph{random} smaller-sized problem, one could approximate the data measures $\mu_i$ by some \emph{deterministic}  approximations supported on a small number of points. One immediate advantage of our random approach is that our error bounds readily apply in case the data measures are random realisations obtained from some experiment (also compare the previous paragraph on population barycenters).   Our results in Section~\ref{sec:bounds} then allow to control the error between the barycenter of the observed data and the barycenter corresponding to the underlying mechanism that generated the data measures.

Though, even if the $\mu_i$'s are considered to be fixed, the sampling approach has significant computational advantages, since generating the empirical measures $\mu_i^S$ can be done in negligible time.  In contrast, a deterministic approximation inevitably involves some sort of optimisation problem.  The related problem of finding the \emph{quantiser}, i.e., the best $S$-supported measure approximating (in $W_p$) a given measure $\mu$, which is extremely difficult even in one dimension (Graf \& Luschgy \cite{graf2007foundations}), provides a serious challenge that to some extent inspired the study of the rate of convergence of the empirical measure (Dereich et al.\ \cite{dereich2013constructive}).  Furthermore, the quantisers will typically not be uniform on their $S$ support points, so computing the barycenter of the quantisers is more complicated than computing the barycenter of empirical measures, for which we can deploy a faster approach due to their uniform weights (see Section~\ref{sec:computation}).

One may also consider the uniform quantisers $\mu^S_{i,unif}$, the best approximation that is uniform on $S$ points (Chevallier \cite{chevallier2018uniform}).  These deterministic approximations converge, in the worst case, at rate $\log S/S$ (faster than $1/\sqrt S$ of $\mu^S_i$), and, being uniform, allow for quick computation of the approximate barycenter (see Section~\ref{sec:computation}). Notably, if we replace the empirical measures by uniform quantisers and invoke \cite[Theorem~3.3]{chevallier2018uniform}, then we can replicate the proof of Theorem~\ref{full_frechetbound} and also obtain a worst case rate of $\log S/S$ for the error of the barycenter of the uniform quantisers in the Fr\'echet functional.

Whilst these stronger bounds are theoretically appealing, there are two significant drawbacks arising from this approach. Firstly, there is little theoretical control on the support of $\mu^S_{i,unif}$; as a consequence, the barycenter of the $N$ uniform quantisers will usually not be supported on the centroid set $\mathcal C$, in which any true barycenter must lie.  Secondly and more importantly, like quantisers with unrestricted weights, finding $\mu^S_{i,unif}$ or a reasonable surrogate thereof is computationally intractable unless $M$ is small. 

Reduction of the support size is not unlike the notion of multi-scale methods for optimal transport (see for instance Gerber and Maggioni \cite{gerber2017multiscale}, Merigot \cite{merigot2011multiscale}, Oberman and Yuanlong \cite{oberman2015efficient}). These start by computing a coarse (i.e., with small support size) approximation of the measures and solving the optimal transport between the coarse versions.  One then uses the resulting solution as a good initial point to the optimal transport between finer approximations, and the procedure is iterated until the full-scale problem is solved.  Thus, these methods speed up the computations substantially by finding a good initial point.  Unfortunately, they will ultimately fail in large-scale problems for which even having a good initial point is insufficient.  One may argue that in such circumstances one should stop the multi-scale approach on a smaller-scale version of the problem (for instance by merging adjacent grid points). This will inevitably cause a loss of information by blurring the data. In contrast, whilst our sampling approach also reduces the size of the problem, we still have a chance to observe any small feature of the data and if we take a sufficient amount of repeats, we will do that without ever having to move to the full-scale problems.

A potentially more suitable comparison might be given by the ``shielding-neighbourhood" approach of Schmitzer \cite{schmitzer2016sparse}.  It allows to solve dense, large-scale problems by solving a sequence of smaller, sparse problems instead. In particular, one can obtain an exact optimal solution of the full-scale problem without the need to solve it directly.  Shielding performs particularly well when the measures have regular supports, such as in the case of images.  A challenge in applying this method is that without strong a-priori assumptions, it is not clear what size exactly the small problems should have, and their construction could be costly, potentially hindering the computational benefits.  In problems having a less regular structure, this method is less efficient.  One might then be willing to compute an inexact, approximate, solution, but it is difficult to tune the size of the smaller problems to match the desired computational effort.  In contrast, our sampling approach allow for exact control of the size of the reduced problems, which enables precise control of the trade-off between desired computational effort and needed accuracy in the results.

\subsection{Outline}
The paper is stuctured as follows. Section~\ref{sec:rot} introduces the randomised optimal transport computation and gives a refined version of the results of \cite{sommerfeld2019optimal} for $N=2$. In Section~\ref{sec:bounds} we extend this to the case $N>2$ and show theoretical guarantees on the quality of the stochastic approximation. Finally, Section~\ref{sec:computation} presents numerical results on simulated and real data, and gives more details on the SUA-algorithm. Some proofs are omitted from the main text, and are developed in an Appendix. An implementation of SUA and a selection of the aforementioned algorithms is available as part of the CRAN R package \verb+WSGeometry+.

\section{Randomised Optimal Transport}\label{sec:rot}
Our approach will be based on reducing the size of the problem by choosing random approximations to the data measures.  To this end, we draw an independent sample $X_1,\dots,X_S$ from each $\mu_i$ and use the \emph{empirical measure}
\[
\mu^S_i(x):=\frac{\# \{k \vert X_k=x\}}{S},\quad \ x \in \mathcal{X}
\]
as a proxy for $\mu_i$. We begin with a discussion for $N=1$ and provide a refined version of Theorem $1$ in~\cite{sommerfeld2019optimal}, which we require in the following.

\begin{theorem}\label{emp_upper}
Let $\mu$ be a measure on a finite space $\mathcal{X}=\{ x_1,\dots ,x_M\}$ endowed with a metric $d$, and let $\mu^S$ be the corresponding empirical measure obtained from a sample of size $S$ from $\mu$.
Then
\begin{align*}
\mathbb{E}[W_p^p(\mu^S ,\mu)]\leq \frac{\diam(\mathcal{X})^p \mathcal{E}}{\sqrt{S}},
\end{align*}
where the constant $\mathcal{E}:=\mathcal{E}(\mathcal{X},p)$ is given by
\begin{align*}
\mathcal{E}:= 
2^{p-1}\underset{q>1,l_{\max}\in \mathbb{N}_0}{\inf} q^p \Bigg[ q^{-(l_{\max}+1)p}M^{\frac{1}{2}} + \left( \frac{q}{q-1} \right)^p \easysum{l=1}{l_{\max}} q^{-lp} \sqrt{\mathcal{N}(\mathcal{X},q^{-l}\diam(\mathcal{X}))}  \Bigg].
\end{align*}
Here $\mathcal{N}(\mathcal{X},\delta)$ denotes the $\delta$-covering number of $\mathcal{X}$, and $\diam(\mathcal X)=\sup_{x,y\in\mathcal X}d(x,y)$.  Moreover, if $p=1$ the factor $\left( \frac{q}{q-1} \right)^p$ can be removed.
When $(\mathcal X,d)\subset (\R^D,\|\cdot\|_2)$, we have for all integers $q\ge2$ that
\begin{align*}
\mathbb{E}\left[W_p^p(\mu,\mu^S)\right] \leq S^{-\frac{1}{2}}D^{\frac{p}{2}} 2^{p-1} \diam(\mathcal{X})^pq^p 	\left\{
		\begin{array}{lll}
			\left( \frac{q}{q-1} \right)^p\frac{q^{p^\prime}}{1-q^{p^\prime}} & \mbox{if } p^{\prime}<0 \\
			 1+\left( \frac{q}{q-1} \right)^pD^{-1}\log_q M& \mbox{if } p^{\prime}=0 \\
			 M^{\frac{1}{2}-\frac{p}{D}}+\left( \frac{q}{q-1} \right)^pq^{p^\prime}\frac{ M^{\frac{1}{2}-\frac{p}{D}}}{q^{p^\prime}-1}& \mbox{if } p^{\prime}>0,
		\end{array}
	\right.
\end{align*}
where $p'=D/2-p$ and the factor $\left( \frac{q}{q-1} \right)^p$ can be omitted if $p=1$.  In particular, using $q=2$ and $p=1$ gives
\begin{align*}
\mathbb{E}\left[W_1(\mu,\mu^S)\right] \leq S^{-\frac{1}{2}}2D^{\frac{1}{2}}\diam(\mathcal{X}) 	\left\{
		\begin{array}{lll}
			1+\sqrt{2} & \mbox{if } D=1\\
			 1+2^{-1}\log_2 M& \mbox{if } D=2 \\
			 M^{\frac{1}{2}-\frac{1}{D}}+\frac{2^{D/2-1} M^{\frac{1}{2}-\frac{1}{D}}}{2^{D/2-1}-1}& \mbox{if } D>2.
		\end{array}
	\right.
\end{align*}
\end{theorem}
Theorem~\ref{emp_upper} improves upon \cite[Theorem~1]{sommerfeld2019optimal} by running the sum over $l$ from 1 instead of zero and allowing a better prefactor than $q^{2p}$. The proof, however, follows similar lines as that of Theorem $1$ in \cite{sommerfeld2019optimal} and is based on constructing an ultrametric tree on $\mathcal X$ (see Kloeckner \cite{kloeckner2015geometric})  that dominates the original metric $d$ (see also Boissard and Le Gouic \cite{boissard2014mean}). The differences from \cite{sommerfeld2019optimal} are given in Appendix~\ref{sec:app}.
 In our proofs in Section~\ref{sec:bounds}, we shall apply the bound in Theorem~\ref{emp_upper} with $p=1$, in which case the improvement over \cite{sommerfeld2019optimal} is a factor of at least 2.

\section{Empirical Wasserstein Barycenter}\label{sec:bounds}
In this section we present our main results. The stochastic approximation in Theorem~\ref{emp_upper} naturally extends from the optimal transport problem to the barycenter problem by computing the barycenter of empirical versions of the $N$ data measures (see  Algorithm \ref{bary_sub_alg}). We point out that at this point step 8 in Algorithm~\ref{bary_sub_alg} could be performed with any solver for the barycenter problem (see Section~\ref{sec:computation}); no specific method is needed at this step in order for Algorithm~\ref{bary_sub_alg} to work. However, we will later utilise a specialised version of the iterative method introduced by Cuturi and Doucet \cite{cuturiFastComputationWasserstein2014}, which we modify to obtain significantly better performance for empirical measures (see Section~\ref{sec:computation}). As in the classical optimal transport setting, we can average our results over multiple runs to reduce variability, leading to $R$ empirical barycenters $\bar{\mu}_1,\dots , \bar{\mu}_R$.  As a final estimator we take the linear average of those empirical barycenters. This is computationally preferable to using $\bar{\mu}_{r^*}={\text{argmin}}_{r\le R} F(\bar{\mu}_r)$, since evaluating the Fr\'echet functional amounts to solving $N$ large-scale optimal transport problems, and is thus computationally costly.  
In fact, we expect the linear mean to have good performance since convexity of the Wasserstein distance extends to the Fr\'echet functional:
\begin{equation}\label{eq:Fconvex}
F^p\left(\frac{1}{R}\easysum{r=1}{R} \bar{\mu}_r \right) \leq \frac{1}{R}\easysum{r=1}{R}F^p(\bar{\mu}_r), \quad p\geq 1.
\end{equation}
For a more detailed discussion on the choice of $R$ as well as the estimator obtained from $R$ repeats, we refer to Subsection~\ref{subsec:R}.

We now provide the theoretical justification for our resampling method, by giving nonasymptotic bounds on the expected error of the empirical barycenter.  These exhibit convergence rate of $S^{-\frac12}$ independently of the dimension.  It will be assumed henceforth that $(\mathcal X,d)$ is a subspace of a geodesic space $\mathcal Y$, so that the set $\mathcal C$ is well-defined. Since any normed vector space is a geodesic space, our results are valid, in particular, when $\mathcal X$ can be embedded in a Euclidean space of arbitrary dimension. The minimum in the Fr\'echet functional is taken over all measures in $\mathcal{P}(\mathcal{Y})$, but can also be equivalently taken only on $\mathcal P(\mathcal C)$. 
In view of the above inequality \eqref{eq:Fconvex}, it suffices to consider the case $R=1$ in the following for simplicity (see Remark~\ref{rem::generalR} for details).
\begin{algorithm}
\caption{Sampling approximation of the Wasserstein barycenter}
\label{bary_sub_alg}
\begin{algorithmic}[1]
\State {Data Measures: $\mu_1,\dots, \mu_N$, sample size $S$, repeats $R$}
\For{$r=1,\dots ,R$}
\For{$i=1,\dots ,N$}
\State {Draw $X^{(i)}_1,\dots ,X^{(i)}_S \sim \mu_i$}
\State {$\mu^S_i= \frac{1}{S}\sum_{k=1}^{S} \delta_{X^{(i)}_k}$}
\EndFor
\State {Solve $\bar{\mu}_r \in \underset{\mu \in \mathcal{P}(\mathcal{C})}\argmin \frac{1}{N}\easysum{i=1}{N}W_p^p(\mu,\mu^S_i)$}
\EndFor

\State {Set $\widehat\mu^*_S=\frac{1}{R}\easysum{r=1}{R} \bar{\mu}_r$}
\Return{Approximation of the empirical Wasserstein barycenter $\widehat\mu^*_S$}
\end{algorithmic}
\end{algorithm}

\begin{theorem}\label{full_frechetbound}
Let $\mu_1,\dots ,\mu_N$ be probability measures on $\mathcal{X}$ and let $\mu_1^{S_1}, \dots ,\mu_N^{S_N}$ be the corresponding empirical measures based on $S_i \in \mathbb{N}$ independent samples. Let $F^p$ and $\widehat F^p_S$, $p\geq 1$, denote the respective $p$-Fr\'echet functionals given by
\begin{align*}
F^p(\mu):=\frac{1}{N}\ \easysum{i=1}{N}\ W_p^p(\mu_i,\mu), \quad \widehat{F}_{S}^p(\mu)=\frac{1}{N}\ \easysum{i=1}{N}\ W_p^p(\mu_i^{S_i},\mu)
\end{align*}
and let $\mu^*$ and $\widehat\mu_{S}^*$ be their respective minimisers.  Then
\begin{align*}
\mathbb{E}[\lvert F^p(\mu^*)-F^p(\widehat \mu_S^*) \rvert ] \leq \frac{2p \ \diam(\mathcal{X})^{p}}{N}\easysum{i=1}{N}\frac{\mathcal{E}(\text{supp}(\mu_i),1)}{\sqrt{S_i}}    ,
\end{align*}
for $\mathcal{E}$ as in Theorem~\ref{emp_upper}. In particular, when $p=2$ and $S_i=S$ for all $i=1,\dots ,N$, we have
\begin{align*}
\mathbb{E}[\lvert F(\mu^*)-F(\widehat \mu_S^*) \rvert ] \leq   \frac{4    \ \diam(\mathcal{X})^2 }{N\sqrt{S}} \easysum{i=1}{N}\mathcal{E}(\text{supp}(\mu_i),1).
\end{align*} 
\end{theorem}
\begin{proof}
In \cite{sommerfeld2018inference} it is shown that for arbitrary measures $\mu,\mu^\prime ,\nu\in \mathcal{P}(\mathcal{X})$ and $p\ge1$, it holds that
\begin{align*}
|W_p^p(\mu,\nu)-W_p^p(\mu^{\prime},\nu) |\leq W_1(\mu,\mu^{\prime})  \ p \ \diam(\mathcal{X})^{p-1}.
\end{align*}
Therefore, for any $\mu\in \mathcal{P}(\mathcal{S})$,
\begin{align*}
\mathbb{E}\left[\lvert F^p(\mu)-\widehat{F}^p_S(\mu)\rvert\right]
&\leq \frac{1}{N}\easysum{i=1}{N}\mathbb{E}\left[\lvert W_p^p(\mu_i^S,\mu)-W_p^p(\mu_i,\mu) \rvert \right] \\
&\leq \frac{1}{N}\easysum{i=1}{N}\mathbb{E} \left[W_1(\mu_i,\mu_i^S)\right]   \diam(\mathcal{X})^{p-1} p \\
&\leq \frac{1}{N}\easysum{i=1}{N} \frac{\mathcal{E}(\mathcal{X}_i,1)   \diam(\mathcal{X})^{p} p}{\sqrt{S_i}},
\end{align*}
where $\mathcal{X}_i:=\text{supp}(\mu_i)$. Since $\mu^*$ and $\widehat \mu_S^*$ are minimisers of their respective Fr\'echet functionals, deduce that
\begin{align*}
\mathbb{E}[\lvert F^p(\widehat{\mu}^*_S)-F^p(\mu^*) \rvert ]&=\mathbb{E}[ F^p(\widehat \mu_S^*)-F^p(\mu^*)]\\ 
&\leq\mathbb{E}[ F^p(\widehat \mu_S^*)-\widehat{F}^p_S(\mu^*)]+ \frac{1}{N}\easysum{i=1}{N} \frac{\mathcal{E}(\mathcal{X}_i,1) \diam(\mathcal{X})^{p} p}{\sqrt{S_i}}\\
&\leq\mathbb{E}[ F^p(\widehat \mu_S^*)-\widehat{F}^p_S(\widehat \mu_S^*)]+ \frac{1}{N}\easysum{i=1}{N} \frac{\mathcal{E}(\mathcal{X}_i,1)   \diam(\mathcal{X})^{p}  p}{\sqrt{S_i}} \\
&\leq \frac{2}{N} \easysum{i=1}{N} \frac{\mathcal{E}(\mathcal{X}_i,1)  \diam(\mathcal{X})^{p} p}{\sqrt{S_i}}.
\end{align*}
\end{proof}
Using the bound on $\mathcal E$ from Theorem~\ref{emp_upper} on the space $\mathcal X=[0,1]^D$ with $p=2$, we obtain
\begin{align*}
\mathbb{E}[\lvert F^2(\mu^*)-F^2(\widehat \mu_S^*) \rvert ] \leq S^{-\frac12}\begin{cases}
8\sqrt 2 (2+\log_2(M))  & D=2\\
183.5 M^{\frac{1}{6}} & D=3.
\end{cases}
\end{align*}

This result gives the rate of approximation in terms of the objective value.  Approximating the optimisers is a more delicate matter, which can be addressed using the linear programming structure of the problem.  Since the barycenters are not necessarily unique, the best we can hope for is to approximate one of them.
\begin{theorem}\label{ws_bound}
Let $\mu_1 ,\dots ,\mu_N$ be probability measures on $\mathcal{X}$ and let $\mu_1^S,\dots ,\mu_N^S$ be empirical measures obtained from $S$ i.i.d.\  samples. Let $\mathbf{B}^*$ be the set of barycenters of the $\mu_i$ and $\mathbf{B}^*_S$ the set of barycenters of the $\mu_i^S$. Then for $p\geq 1$
\begin{align*}
\mathbb{E}\left[ \underset{\widehat \mu_S^* \in \mathbf{B}_S^*}{\sup} \underset{\mu^* \in \mathbf{B}^*}{\inf} W^p_p(\mu^*,\widehat \mu_S^*) \right]\leq  \frac{p \bar{\mathcal{E}}  \ \diam(\mathcal{X})^{p} }{C_{P}}S^{-\frac12},
\end{align*}
where $\bar{\mathcal{E}}= \sum_{i=1}^{N}\mathcal{E}(\text{supp}(\mu_i),1)/N$ and $C_P$ is a strictly positive constant given by
\begin{align*}
C_{P}:=C_{P}(\mu_1,\dots,\mu_N):=(N+1)\diam(\mathcal{X})^{-p}\underset{v \in V\backslash V^*}{\min} \ \frac{c^Tv-f^*}{d_1(v,\mathcal M)},
\end{align*}
where $V$ is the vertex set of the linear programming formulation of the barycenter problem \eqref{eq:linProgBary}, $V^*$ is the subset of optimal vertices, $c$ is the cost vector of the program, $f^*$ is the optimal value, $\mathcal M$ is the set of minimisers of the problem \eqref{eq:linProgBary}, and $d_1(v,\mathcal M)=\inf_{x\in M}\|v-x\|_1$.
\end{theorem} 
\begin{remark}
Using the convexity of $P$ and the linearity of the objective function, one can show that the minimum in $C_P$ is attained at a vertex $v\in V\setminus V^*$ which is adjacent to some $v^*\in V^*$.
\end{remark}
\begin{remark}
A stronger version of Theorem~\ref{ws_bound} can be shown, if one uses total variation instead of the $p$-Wasserstein distance. This can be achieved by skipping the last inequality of the proof.
\end{remark}
\begin{remark}
Theorem~\ref{ws_bound} can be generalised, using (\ref{eq:Fconvex}), to a general $R$, in which case the statement becomes slightly more complicated, however. For $r=1,\dots ,R$ let $\textbf{B}_{S,r}^*$ be the set of barycenters of the $N$ empirical measures from the $r$-th repeat. The set $\textbf{B}_S^*$ needs to be replaced by the set
\begin{align*}
\left\{ \frac{1}{R} \easysum{r=1}{R} \mu_r \ \Bigg\vert \ \mu_r \in  \textbf{B}_{S,r}^* \right\}
\end{align*}
of measures that can be obtained as a linear average of barycenters from the $R$ repeats. The upper bound remains the same.
\label{rem::generalR}
\end{remark}
\begin{remark}
The rate $S^{-\frac12}$ is optimal. This can already be seen in the case $N=1$, where the Fr\'echet functional is simply the $p$-th power of the $p$-Wasserstein distance between a measure $\mu$ and its empirical version $\mu^S$.  For a finitely supported and nondegenerate measure $\mu$, this has a lower bound scaling as $S^{-\frac12}$ (e.g., Fournier and Guillin \cite{fournier2015rate}), implying optimality of our rate. For a concrete example let $\mu=\frac{1}{2}(\delta_0 + \delta_1)$. By construction we have for some random variable $K\sim \text{Bin}(S,1/2)$ that
\begin{align*}
\mathbb{E}\left[W_p^p(\mu,\mu^S)\right]=\mathbb{E} \left\vert \frac{1}{2}-\frac{K}{S} \right\vert = S^{-1} \lvert \mathbb{E} [K]-K \rvert\geq S^{-1}2^{-\frac{1}{2}}\sqrt{S/4} = \sqrt 2S^{-\frac{1}{2}}/4,
\end{align*}
where the inequality follows from the properties of the mean absolute deviation of the binomial distribution (Berend and Kontorovich \cite{berend2013sharp}).
\end{remark}
The proof of Theorem~\ref{ws_bound} requires two auxiliary, geometric results that may be of independent interest.  To streamline the presentation, their proofs are given in Appendix~\ref{sec:app}.

For a nonempty $P\subseteq\mathbb R^L$ and $x\in \mathbb R^L$, define $d_1(x,P)=\inf_{z\in P}\|x-z\|_1$.  When $P$ is closed, the infimum is attained, since it can be taken on the compact set $P\cap\{y:\|y-x\|\le d_1(x,P)+1\}$.
\begin{lemma} \label{poly_affine}
Let $P \subseteq \mathbb{R}^L$ be a nonempty polyhedron and let $x,y \in \mathbb{R}^L$.  Then the function
\begin{align*}
g(t) = d_1(x+ty,P)=\min_{z\in P} \lVert x+ty-z \rVert_1,\qquad t\in \R
\end{align*}
is convex, Lipschitz and piecewise affine.
\end{lemma}
An illustration of Lemma~\ref{poly_affine} in $\R^2$ is given in Figure~\ref{fig:poly_affine}. 
\begin{figure}
\includegraphics[width=\textwidth]{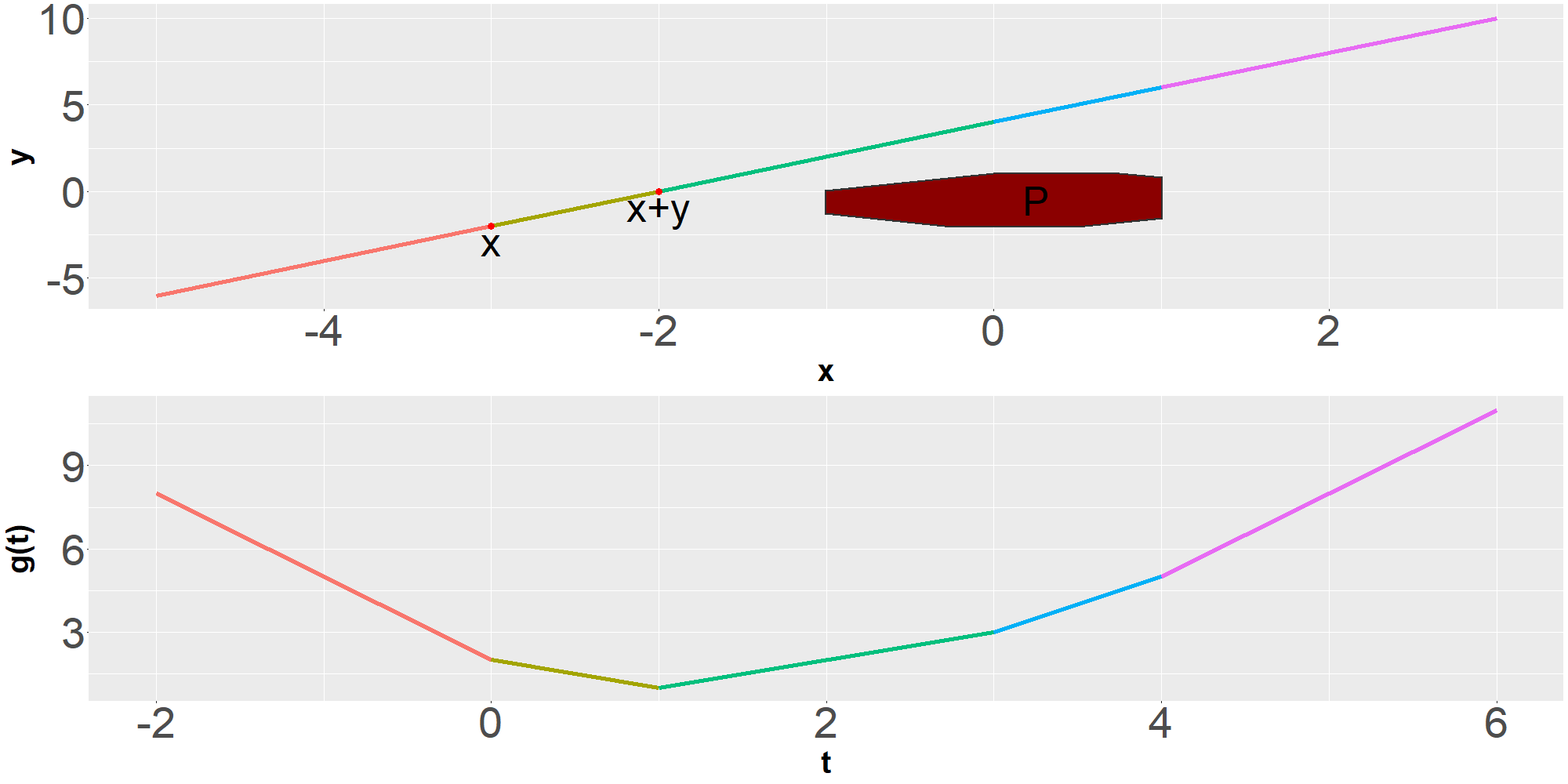}
\caption{\textbf{Top:} the polytope $P$ and the vectors $x=(-3,-2)$ and $y=(1,2)$.  \textbf{Bottom:} the function $g(t)$.  The different colours correspond to segments on which $g$ is affine.  On the segment $t\in (1,3)$ the minimiser in $d_1(x+ty,P)$ is not unique.}
\label{fig:poly_affine}
\end{figure}
\begin{figure}
\includegraphics[width=\textwidth]{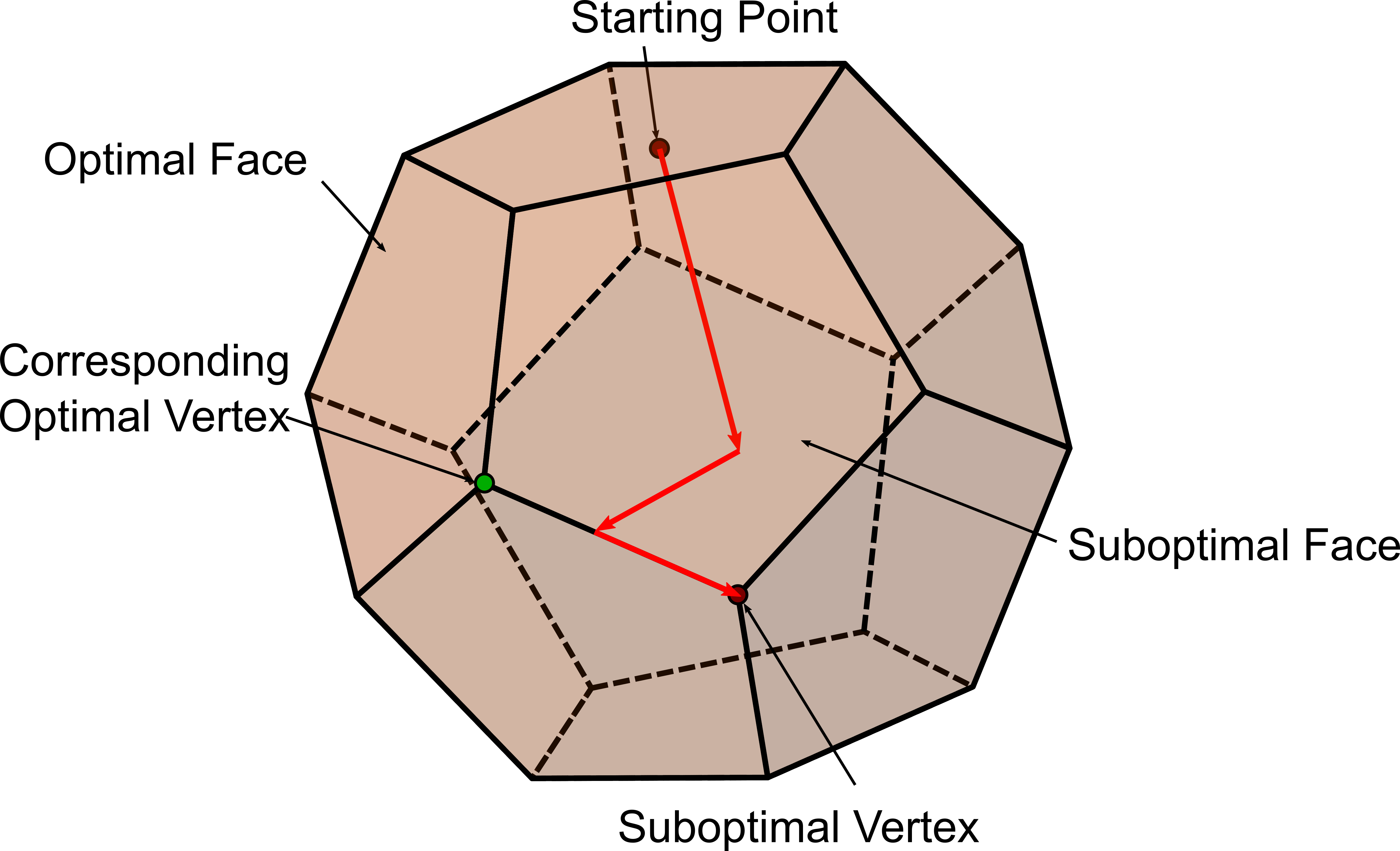}
\caption{Sketch of a path of on which $\psi$ is nonincreasing. We start at a point in the interior of $P$, then move to a face/subface until we hit a vertex.}
\label{fig:vertices}
\end{figure}
\begin{lemma}\label{slope_bound}
Let $F^p$ be the Fr\'echet functional corresponding to $\mu_1,\dots ,\mu_N \in \mathcal{P}(\mathcal{X})$.  Then for any $\mu \in \mathcal P(\mathcal{C})$ there exists a $\mu^* \in {\argmin}_{\mu } F^p(\mu)$ such that
\begin{align*}
F^p(\mu)-F^p(\mu^*) \geq 2C_{P} W_p^p (\mu,\mu^*),
\end{align*}
where $C_{P}$ is the constant from Theorem~\ref{ws_bound}.
\end{lemma}
\begin{proof}[Sketch of proof of Lemma~\ref{slope_bound}]
Let $P\subset \mathbb{R}^L$ be the feasible polytope corresponding to \eqref{eq:linProgBary} and let $f^*$ be the optimal objective value of this linear program. 
To each $\mu \in \mathcal{P}(\mathcal{\mathcal{C}})$ corresponds a $\pi \in P$ such that $c^T\pi=F(\mu)$. 
Fix an element $\pi^* \in \argmin d_1(\pi,\mathcal{M})$, from which we can construct a minimiser $\mu^*$ of $F$. It holds that
\[
F^p(\mu)-F^p(\mu^*)
=c^T \pi - c^T \pi^* =\|\pi-\pi^*\|_1 \frac{c^T \pi - f^*}{\|\pi-\pi^*\|_1}\ge \|\pi-\pi^*\|_1\psi(\pi),
\]
where
\[
\psi (\pi)= \frac{c^T\pi-f^*}{d_1(\pi,\mathcal M)},
\qquad \pi\in P\setminus \mathcal M.
\]
One can now minimise $\psi$ over all elements of $\mathcal{P}$ and show that the minimum is positive and attained at a vertex of the polytope. The main idea of the proof is that for each point $\pi  \in \mathcal{P}$ which is not a vertex, one can find a direction in $\mathcal{P}$ in which one can move without increasing the value of $\psi$. We can then move into this direction until we hit a face/subface of $\mathcal{P}$, at this point we can then construct a new direction along which we can move unless we have hit a vertex. A sketch of this in three dimensions can be seen in Figure~\ref{fig:vertices}. The value of $\psi$ can be controlled using Lemma~\ref{poly_affine}, as it is a ratio of an affine and a piecewise affine convex function.   Finally, we can use the structure of the linear program to bound $\|\pi-\pi^*\|_1$ by the Wasserstein distance between the corresponding measures $\mu$ and $\mu^*$. 
\end{proof}
\begin{proof}[Proof of Theorem~\ref{ws_bound}]
Invoking Lemma~\ref{slope_bound} with $\mu\in \mathbf{B}_S^*$ and applying Theorem~\ref{full_frechetbound} yields
\begin{align*}
\frac{2p \bar{\mathcal{E}}\diam(\mathcal{X})^{p} }{\sqrt{S}} \geq \mathbb{E}\left[F^p(\widehat \mu_S^*)-F^p(\mu^*) \right]\geq \mathbb{E}\left[2C_{P}   \underset{\widehat \mu_S^* \in \mathbf{B}_S^*}{\sup} \underset{\mu^* \in \mathbf{B}^*}{\inf}W_p^p (\widehat \mu_S^*,\mu^*)\right].
\end{align*}
Thus
\begin{align*}
\frac{p \bar{\mathcal{E}} \ \diam(\mathcal{X})^{p} }{C_{P}}S^{-\frac12} \geq \mathbb{E}\left[ \underset{\widehat \mu_S^* \in \mathbf{B}_S^*}{\sup} \underset{\mu^* \in \mathbf{B}^*}{\inf}W^p_p(\mu^*,\widehat \mu_S^*)\right].
\end{align*}
\end{proof}
The constants $C_P$ and $\bar{\mathcal E}$ are scale invariant, i.e., they remain the same if the metric $d$ is multiplied by a positive constant.  Finding a more explicit lower bound for $C_P$ in specific cases (such as measures supported on a regular grid) is an important task for further research.

\begin{remark}
If the number of measures $N$ increases whilst $S$ is kept fixed, then the total number of observations $NS$ increases linearly in $S$.  From a sample complexity perspective, one might be interested in fixing the total sample size $L=\sum_{i=1}^{N}S_i$.  In view of Theorem~\ref{full_frechetbound}, consistency of the Fr\'echet value is achieved if all the $S_i$'s diverge to infinity.  In the ``uniform" case, where $S_i=L/N$, this translates to requiring $N/L\to0$.
\end{remark}

\section{Computational Results}\label{sec:computation}
\subsection{The SUA-Algorithm}
Over the past years many algorithms that allow to compute approximated Wasserstein barycenters have emerged. Most of them focus on either the regularised problem, or on computing the barycenter on a fixed, pre-specified support set.  For reasons discussed in the introduction, we aim to solve the exact, unregularised problem.  In the context of stochastic sampling, methods that consider a common, fixed support are inappropriate, as the resampling inevitably leads to empirical measures with different supports.  Thus, we employ a modification of the iterative, alternating algorithm proposed by Cuturi and Doucet  \cite{cuturiFastComputationWasserstein2014}. Alternating between position and weight updates for general measures creates a massive computational burden. However, there is an interesting empirical observation that renders the algorithm of \cite{cuturiFastComputationWasserstein2014} particularly appealing in the context of stochastic sampling.

Suppose that the measures $\mu_1,\dots ,\mu_N$ are all uniform on the same number of points.  In the notation of Section~\ref{sec:foundations}, we assume that $M_1=M_2\dots =M_N=M$ and $b_k^i=1/M$ for all $i=1,\dots ,N$ and all $k=1,\dots ,M$.  When $N=2$, then we can use the Birkhoff-von Neumann theorem to show that the barycenter of these two measures is also a uniform measure on $M$ points, regardless of the cost function.  If $N\ge 3$ (and $M\ge 3$), then the feasible polyhedron of the barycenter linear program \eqref{eq:linProgBary} has vertices that are not uniform.  Consequently, it is not clear whether the minimiser is also uniform on $M$ points (see also Lin et al.\ \cite{lin2020computational}).  However, we have empirically observed that for the squared Euclidean cost, there is a barycenter that is uniform on $M$ points.  This has been tested on an overwhelming amount of simulations in different scenarios and, with no exception, the barycenter was always uniform.  We do not have a formal proof of this uniformity conjecture, which poses an interesting question for further research. However, in view of the empirical evidence we shall assume in the following that it is practically valid, and explore its computational consequences.

Our observations suggest two immediate improvements to the alternating procedure. Firstly, we no longer need to perform weight updates: they can be chosen uniform. Secondly, we can reduce the number of support points of the candidate barycenter from $MN-N+1$ to $M$. This would have probably been done in practice anyway to reduce computational strain, as, for instance, using a fixed support approximation of the barycenter essentially implies choosing the support size of the barycenter to be $M$ and not performing positions updates. Our per-iteration computational cost is therefore comparable to that of a fixed support barycenter. However, based on this conjecture this support size is actually optimal. Finally, we replace the one-step Newton-type update of \cite{cuturiFastComputationWasserstein2014} by a proper subgradient descent on the positions of the candidate measure.  The procedure can be seen in Algorithm~\ref{SUA}. Note that this problem is non-convex, so, in order to avoid local minima, we give the algorithm a ``warmstart" by performing a fixed number of stochastic subgradient descent steps to obtain an initial point for the genuine (i.e., non-stochastic) subgradient descent.  The left panel of Figure~\ref{HS} shows the advantages of this small modification, which comes at little computational cost, but provides a noticeable improvement. Here, performance is measured in terms of the (empirical) relative error in the Fr\'echet functional, namely
\begin{align*}
\frac{F(\widehat\mu^*_S)-F(\mu^*)}{F(\mu^*)}.
\end{align*}
Without the warmstart our runs provided a median relative error of $5.2\%$, which reduced to $0.8\%$ with a warmstart. Over half of the runs with warmstart yielded a relative error under $1\%$, which is a good indication that we can obtain good results using very few restarts. However, even without the warmstart, we typically obtain relative errors in the single-digit percent points, which is still tolerable in many applications.

Whilst the measures $\mu_1,\dots,\mu_N$ are typically not uniform, the uniform approximation can be applied to their empirical versions $\mu_1^S,\dots,\mu_N^S$; the algorithm is still applicable if some points are duplicate by taking some of the rows of some $Y^i$ to be the same.  Thus, the SUA-method reduces the problem size by replacing $M_i$ with $S$, whilst at the same time enforcing a setting in which a much faster method can be employed.  This two-fold gain pushes down the computation time by orders of magnitude, and provides the basis for our simulation results in the following section. 
\begin{figure}
    \centering
    \begin{subfigure}{0.48\textwidth}
            \includegraphics[width=\textwidth]{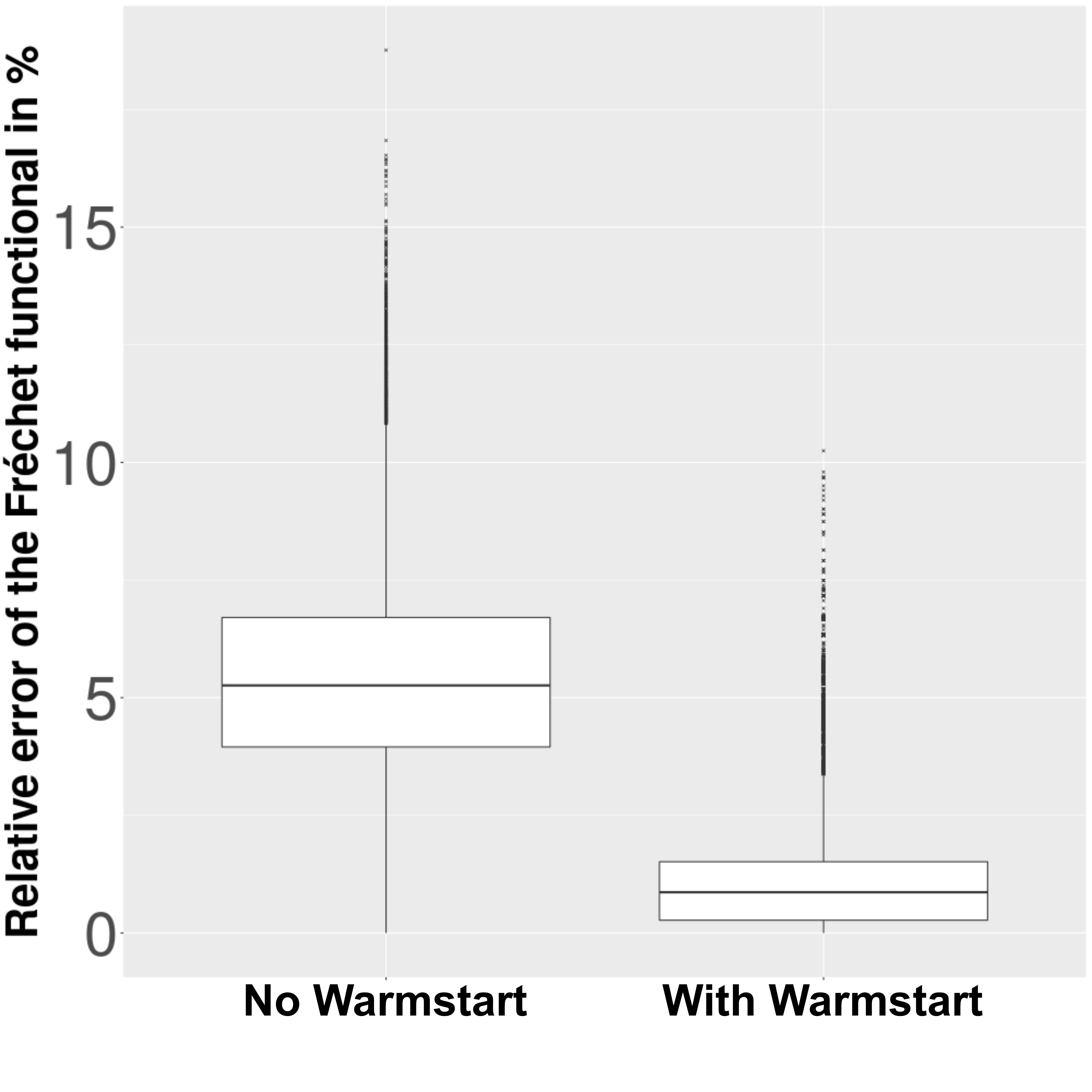}
    \end{subfigure}
        \begin{subfigure}{0.48\textwidth}
    \includegraphics[width=\textwidth]{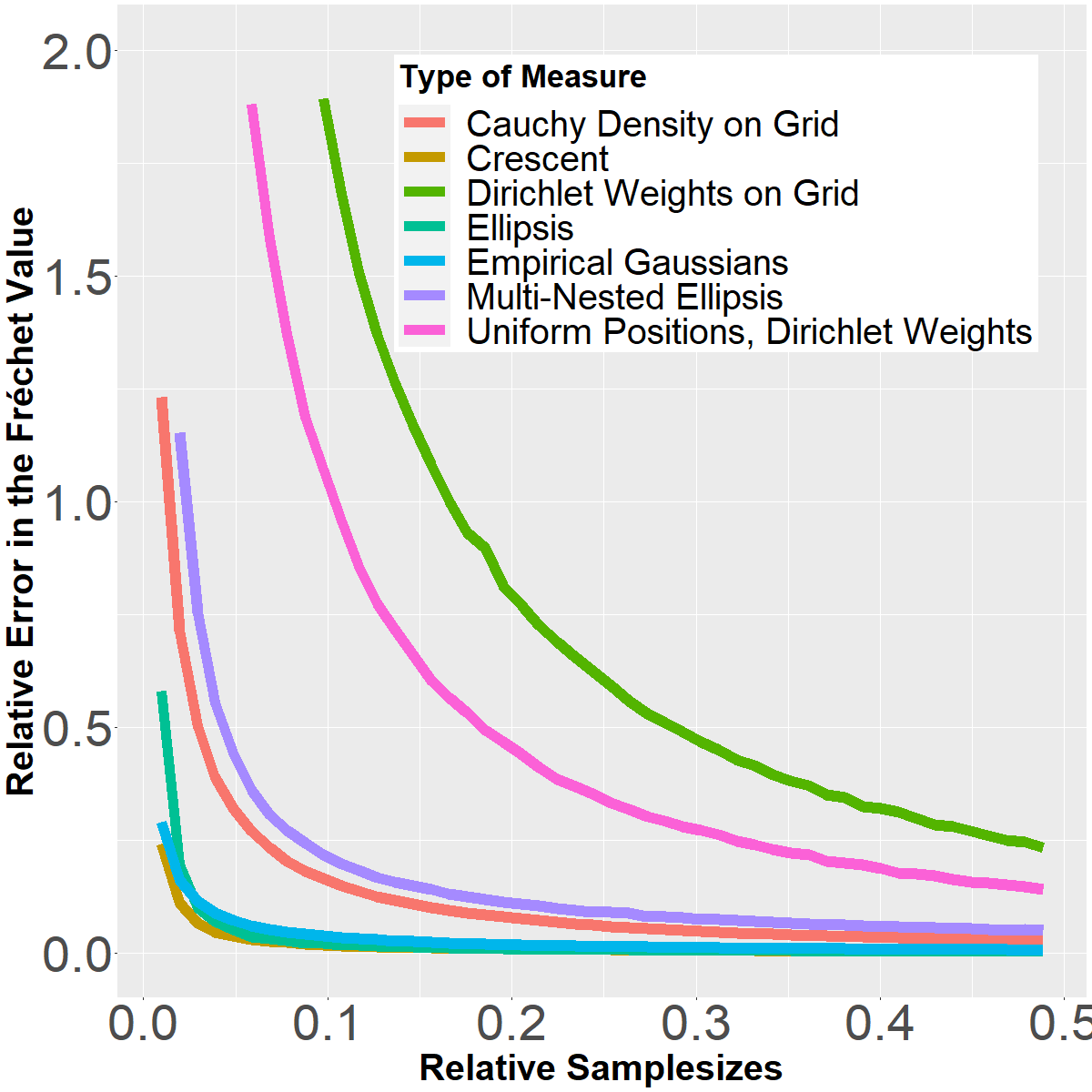}
    \end{subfigure}
    \caption{\textbf{Left:} Relative Fr\'echet error from $50000$ runs, on $N=4$ measures with $M=20$ uniformly random positions with weights $1/M$, with and without warmstart, respectively. \textbf{Right:} Simulated expected error of the Fr\'echet functional of $N=20$ measures with $M=1024$ points in $\mathbb{R}^2$ based on $100$ repeated runs for each value of $S$. Example measures from the corresponding dataset and example barycenters for some values of $S$ can be seen in Figure~\ref{fig:2ddata} and Figure~\ref{ImageBarys}, respectively.}
        \label{HS}
\end{figure}
\begin{algorithm}
\caption{Stochastic Uniform Approximation (SUA)--algorithm for the Wasserstein barycenter}
\label{SUA}
\begin{algorithmic}[1]
\State {Data Measures: $\mu_1,\dots, \mu_N$, sample size $S$, repeats $R$, initial position matrix $X^0\in \mathbb{R}^{S \times d}$, sequence of stepsizes $(\alpha_n)_{n \in \mathbb{N}}$, $\alpha_n>0 \ \forall \ n\in\mathbb{N}$.}
 \State {$\Pi = \{ T\in \mathbb{R}^{S \times S} \quad\vert \quad T \textbf{1}_{S}=\textbf{1}_S^TS^{-1}=T^T\textbf{1}_S^T\}$}
\For{$r=1,\dots ,R$}
\For{$i=1,\dots ,N$}
\State{Draw $X^{(i)}_1,\dots ,X^{(i)}_S \sim \mu_i$}
\State {$\mu^S_i= \frac{1}{S}\sum_{k=1}^{S} \delta_{X^{(i)}_k}$}
\EndFor
\State{$Y^{i}=\text{supp}(\mu_i^S)$}

\While{not converged}
\For{$i=1\dots ,N$}
\State{$T_i = \underset{T \in \Pi}{\argmin}$ $\easysum{k=1}{S}\easysum{j=1}{S}T_{kj}\lVert X_k-Y_j^i \rVert_2^2$}
\State{$V_i = 2(X-T_iY^{(i)})$}
\EndFor
\State{$X^{n+1} = X^{n}-\frac{\alpha_n}{N}\easysum{i=1}{N}V_i$}
\State{$n = n+1$}
\EndWhile
\State{$\bar{\mu}_r=\frac{1}{S}\easysum{k=1}{S}\delta_{X_k^{n}}$}
\EndFor
\State{Set $\widehat\mu^*_S=\frac{1}{R}\easysum{r=1}{R} \bar{\mu}_r$}
\Return{Approximation of the empirical Wasserstein barycenter $\widehat\mu^*_S$}
\end{algorithmic}
\end{algorithm}

\subsection{Simulations}
\begin{figure}
    \centering
    \begin{subfigure}{0.48\textwidth}
            \includegraphics[width=\textwidth]{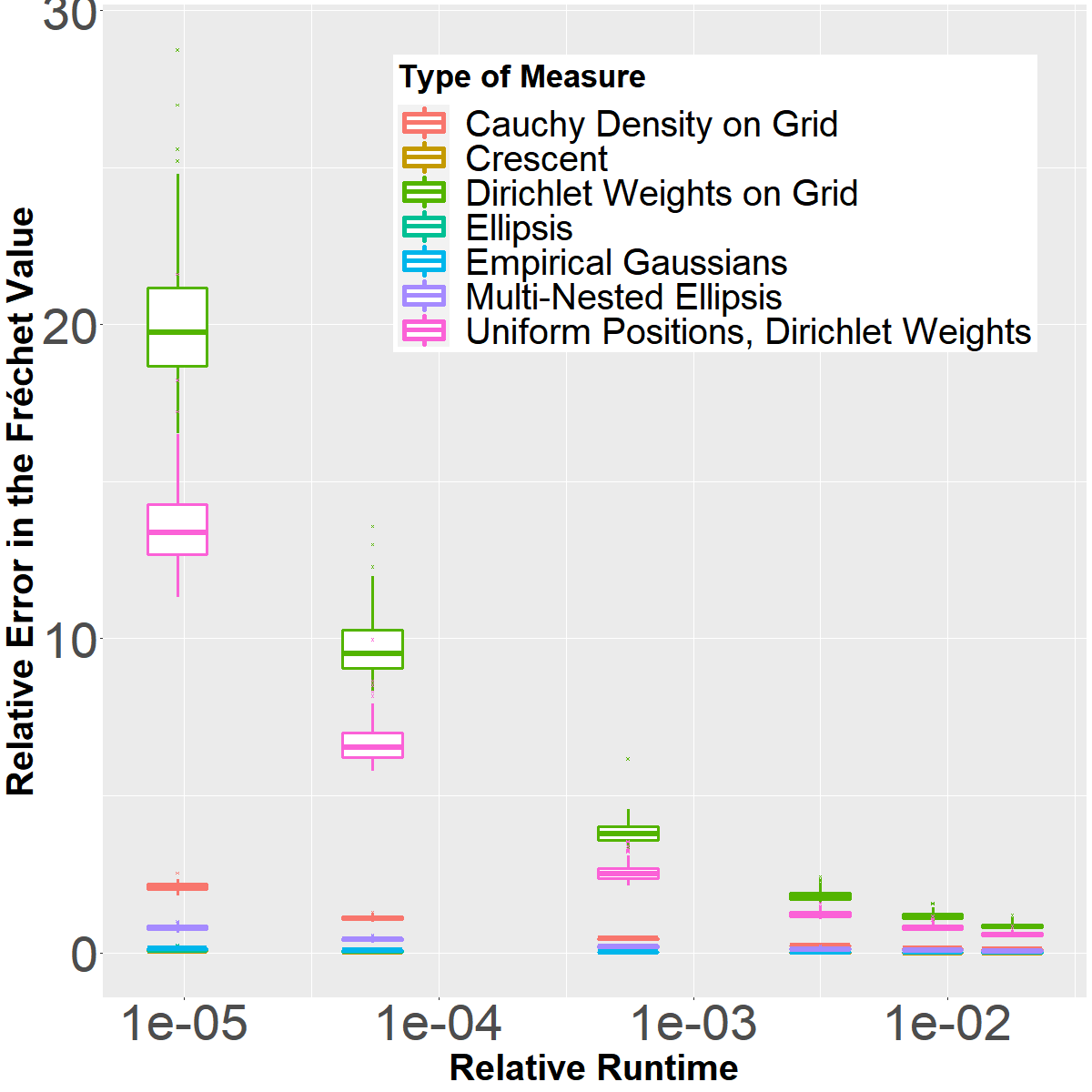}
    \end{subfigure}
        \begin{subfigure}{0.48\textwidth}
    \includegraphics[width=\textwidth]{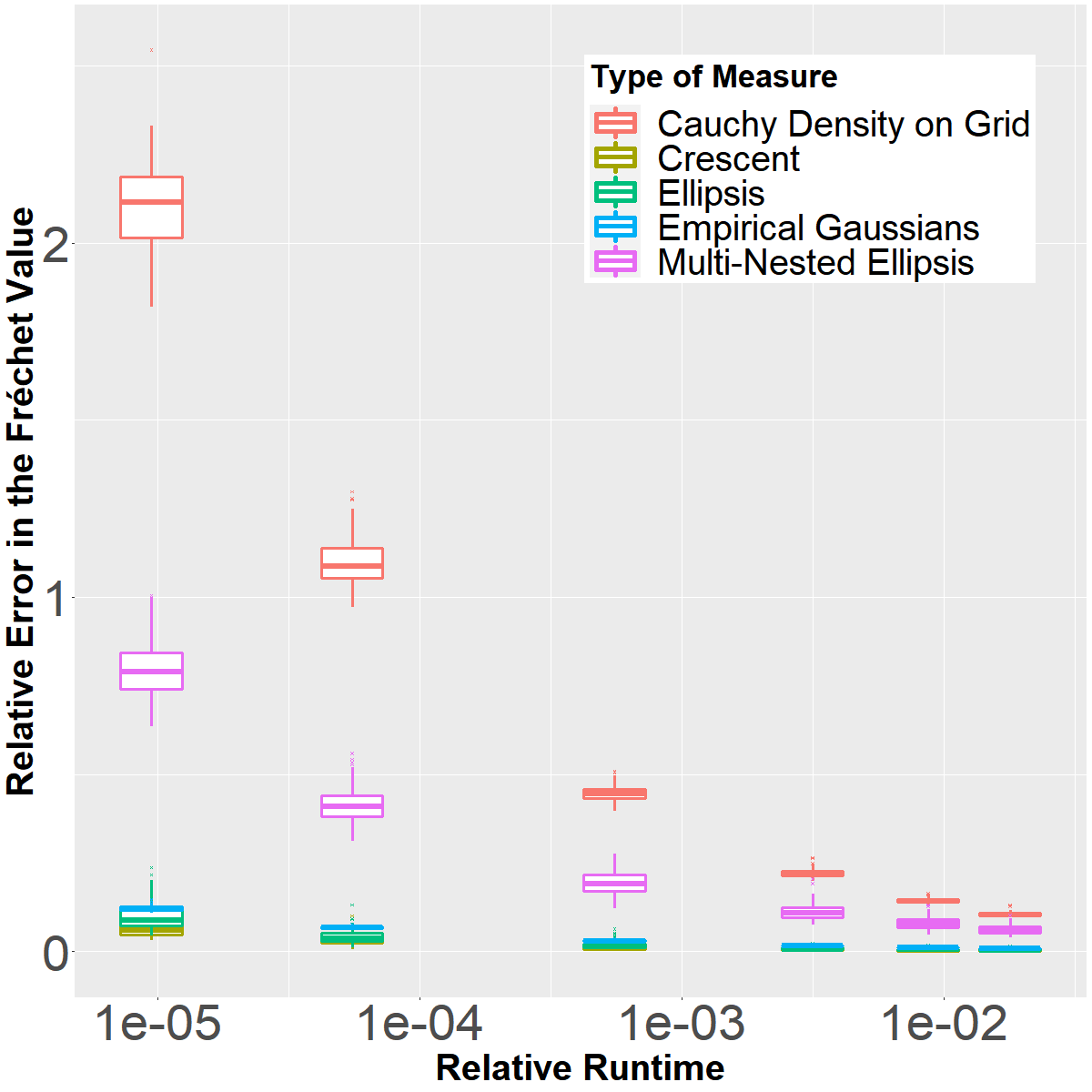}
    \end{subfigure}
    \caption{Simulated expected error of the Fr\'echet functional of $N=10$ measures with $M=4096$ points in $\mathbb{R}^2$ based on $100$ repeated runs for each value of $S$ compared to their relative runtime compared to the original dataset. On the right panel the two outlier distributions have been removed.}
        \label{sdep}
\end{figure}
\begin{figure}
    \centering
    \includegraphics[width=\textwidth]{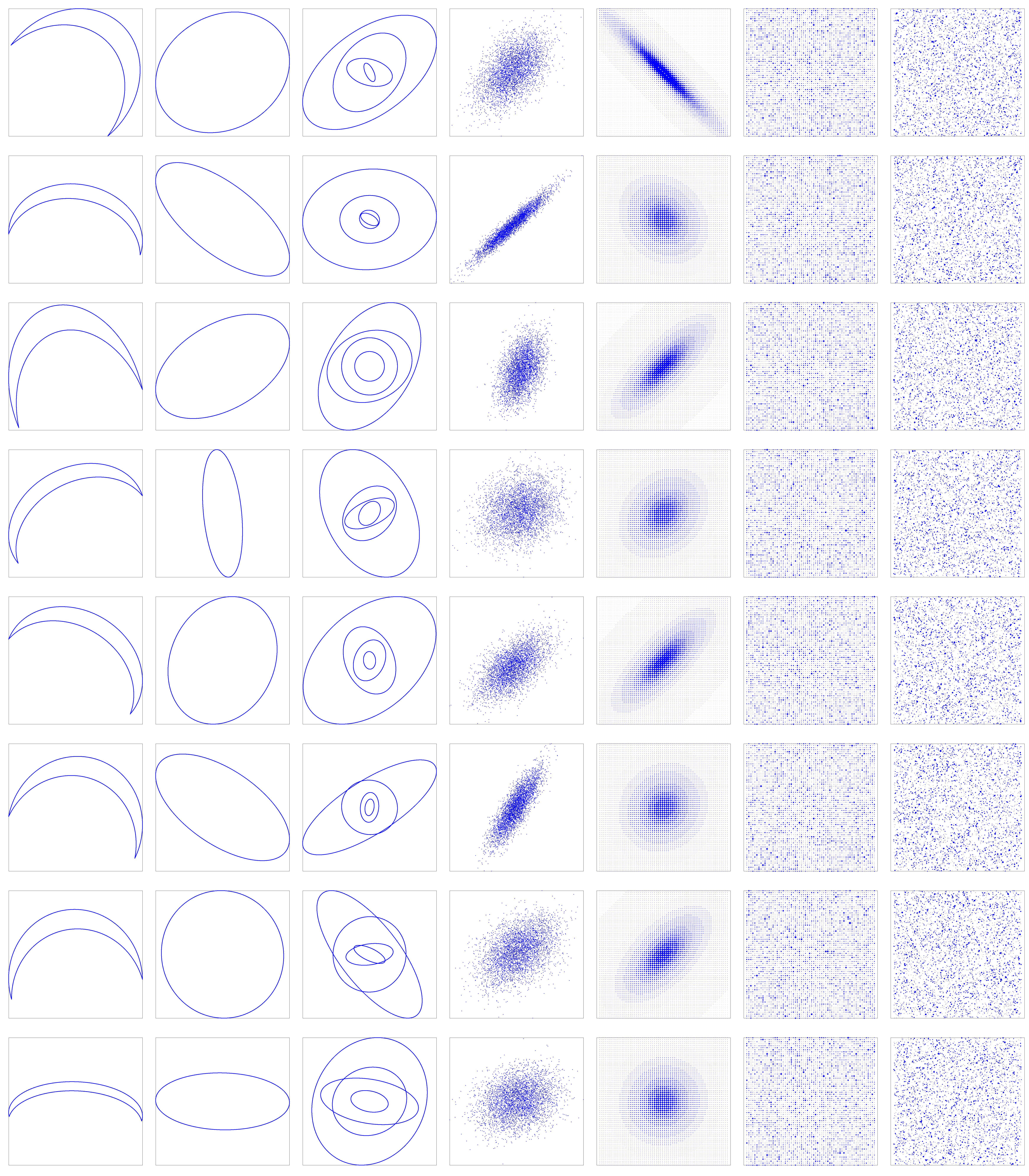}
    \caption{1.Column: Discretised crescents with uniform weights, 2.Column: Discretised ellipses with uniform weights, 3.Column: Discretised nested ellipses with uniform weights, 4.Column: Gaussian positions with uniform weights, 5.Column: Discretised Cauchy density on a grid, 6.Column: Dirichlet weights on a grid, 7.Column: Dirichlet weights on uniform positions. In each column the measures above the line are excerpts from the used dataset of $N=100$ measures with $M=4096=64^2$ points.}
\label{fig:2ddata}
\end{figure}
\begin{figure}
    \centering
    \includegraphics[width=\textwidth]{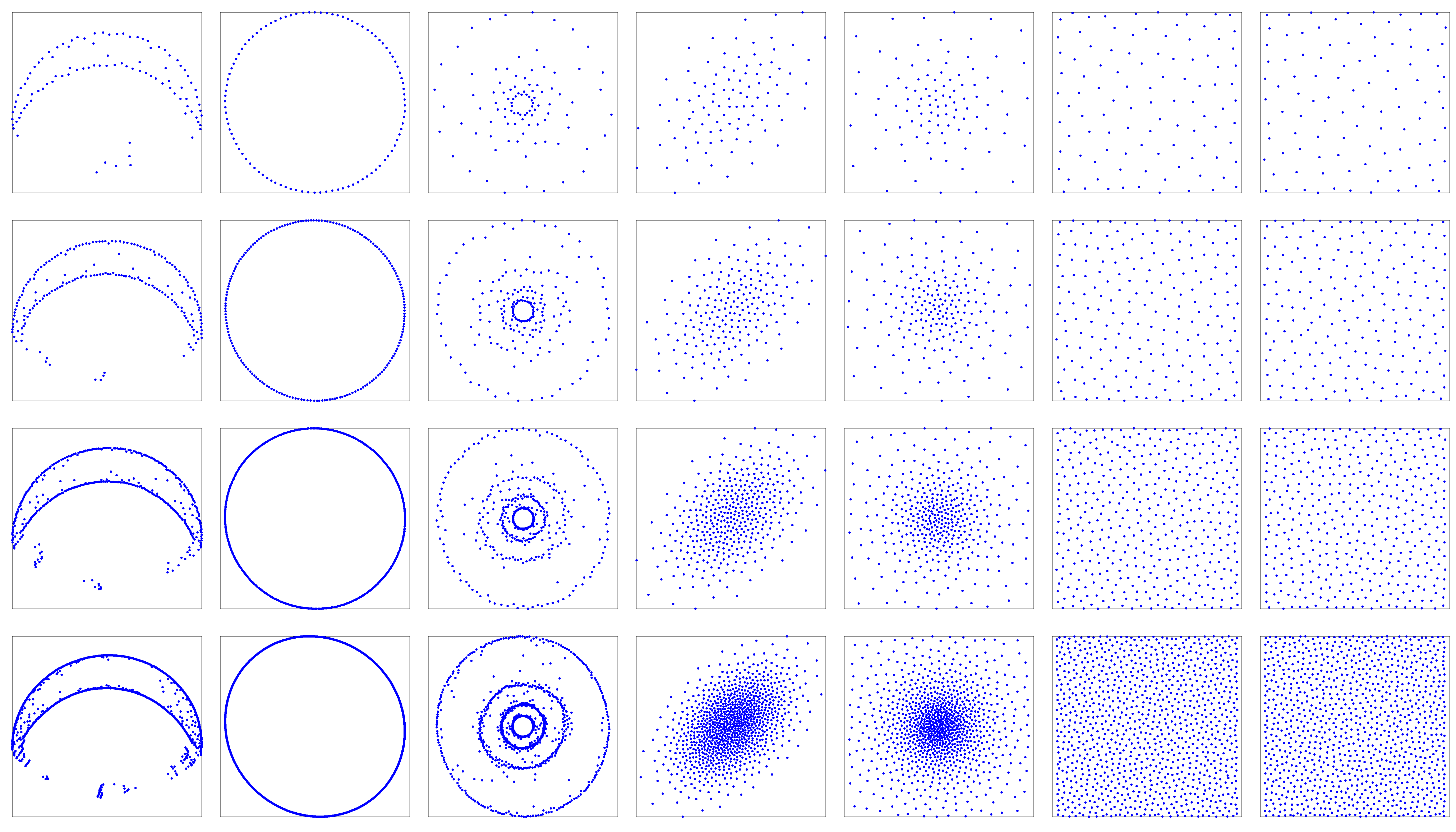}
  \caption{The randomised approximations of the barycenters of the data in Figure~\ref{fig:2ddata} obtained from the SUA-algorithm. The first row shows empirical barycenters based on $S=100$ samples, the second for $S=200$, third for $S=400$ and the fourth for $S=1000$. The computation of one barycenter in the last row took about 40 minutes on a single core of an i7-7700k.}
          \label{ImageBarys}
\end{figure}
\begin{figure}
    \centering
    \includegraphics[width=0.6\textwidth]{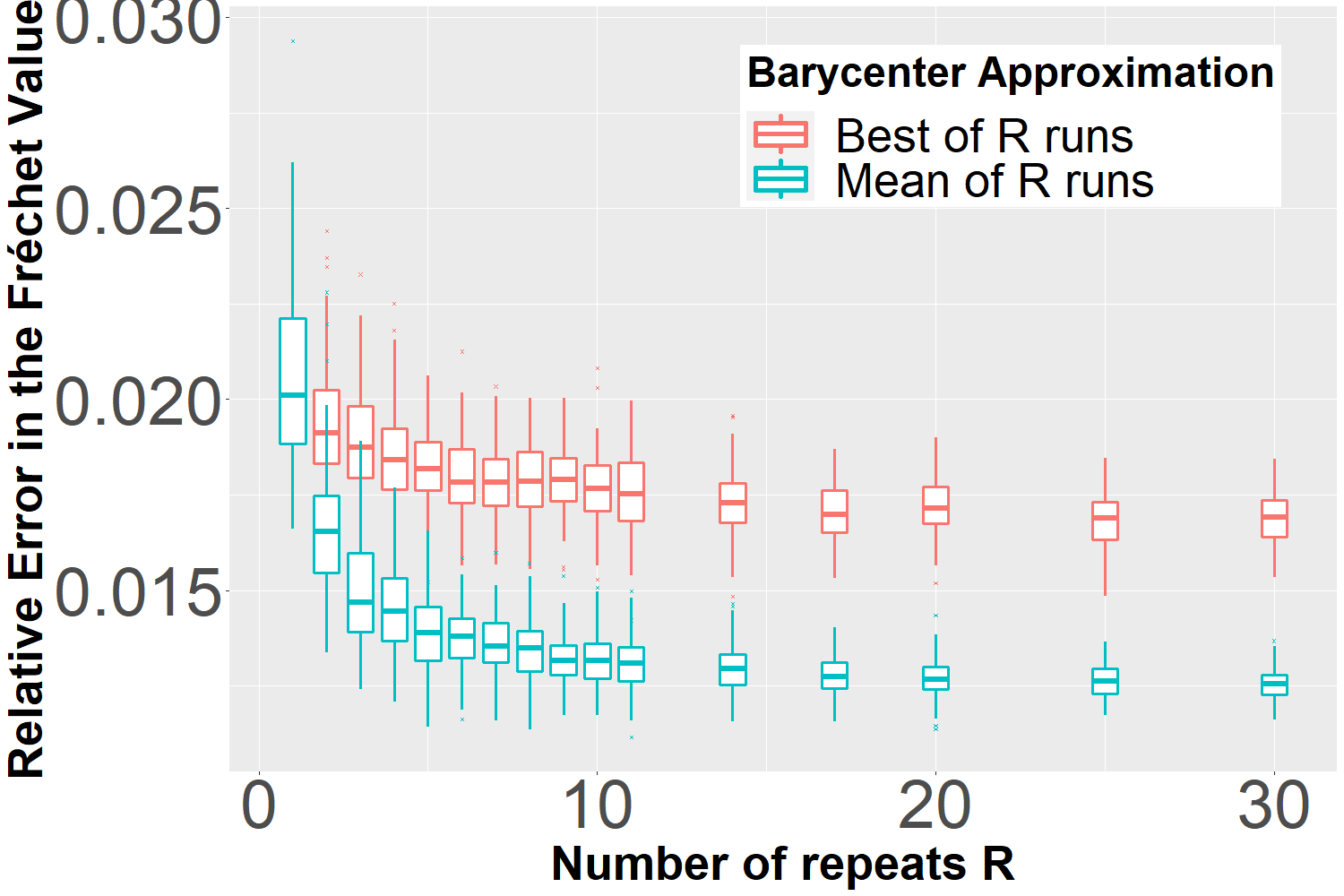}
    \caption{The dependence of the empirical barycenter approximation on the number of repeats $R$. Here, the data measures are $N=20$ multi-nested ellipses supported on $M=100$ support points each. The sample size $S$ was chosen to be $33$. For each $R$, $100$ runs have been performed.}
    \label{meanvsbestS33}
\end{figure}
In this subsection we present empirical results on the decay of the approximation error as the resample size $S$ grows, for some classes of measures on $\mathbb R^2$. {As we discuss the impact of the number of repeats $R$ separately in the following subsection, we set $R=1$ whilst analysing the behaviour with respect to $S$.}

The right panel of Figure~\ref{HS} shows the simulated, expected error as a function of $S$, for the datasets showcased in Figure~\ref{ImageBarys}. Noticeably, performance greatly varies between different types of datasets. For simple geometric structures such as ellipses, crescent shapes, or Gaussian samples, we obtain the best results; for $S$ of the order of $10\%$ of the original data size $M$ we obtain expected approximation errors of under $5\%$. For more complicated structures, such as the nested ellipses and the Cauchy density on a grid, we still obtain decent empirical errors. For $S\approx 0.1M$ we observe an error of under $20\%$ in the Fr\'echet functional, whilst the runtime is improved by a factor of about 300. On the unstructured data, such as Dirichlet weights on a grid/uniform support points, we observe the worst results. For $10\%$ support size we still observe errors upwards of $100\%$. As barycenters are mostly used in settings where the underlying measures have some form of geometric structure, these results indicate good potential performance of our stochastic barycenter approximation. We point out that the variance of the observed Fr\'echet values was generally quite small, as can be seen in Figure~\ref{sdep}, and decreases with increasing sample size. {This suggests that one might be best advised to increase the number of samples $S$ instead of increasing the number of repeats $R$ in Algorithm~\ref{bary_sub_alg}. In fact, this coincides with the observation by Sommerfeld et al.\ \cite{sommerfeld2019optimal} for the empirical Wasserstein distance itself as well as our own findings in the next subsection.}  It is also noteworthy that the curves at the right panel of Figure~\ref{HS} suggest that doubling the sample size approximatively halves the approximation error, indicating linear decay in $S$ instead of the rate $S^{-\frac12}$. This suggests that, whilst asymptotically we cannot perform better than this rate, we can achieve even better results for small sample sizes.

Since the convergence rate in Theorem~\ref{full_frechetbound} is independent of the dimension of the ambient space of the data, it is natural to extend our simulations from the plane to higher dimensions. In Figure~\ref{3d_bary} we provide some examples of datasets in $\mathbb R^3$.  The approximation indeed performs well, allowing to compute rather good approximate barycenters when $S$ is less than $5\%$ of the original size $M$. Whilst we were no longer able to approximate the barycenter of these measures directly due to memory constraints, extrapolating from smaller datasets suggests that the corresponding computation would take well over a week to complete. On the contrary, our stochastic approximation took around $90$ minutes on a single core of an i7-7700k and used about 800MB of RAM with our implementation.

\begin{figure}
    \centering
    \includegraphics[width=\textwidth]{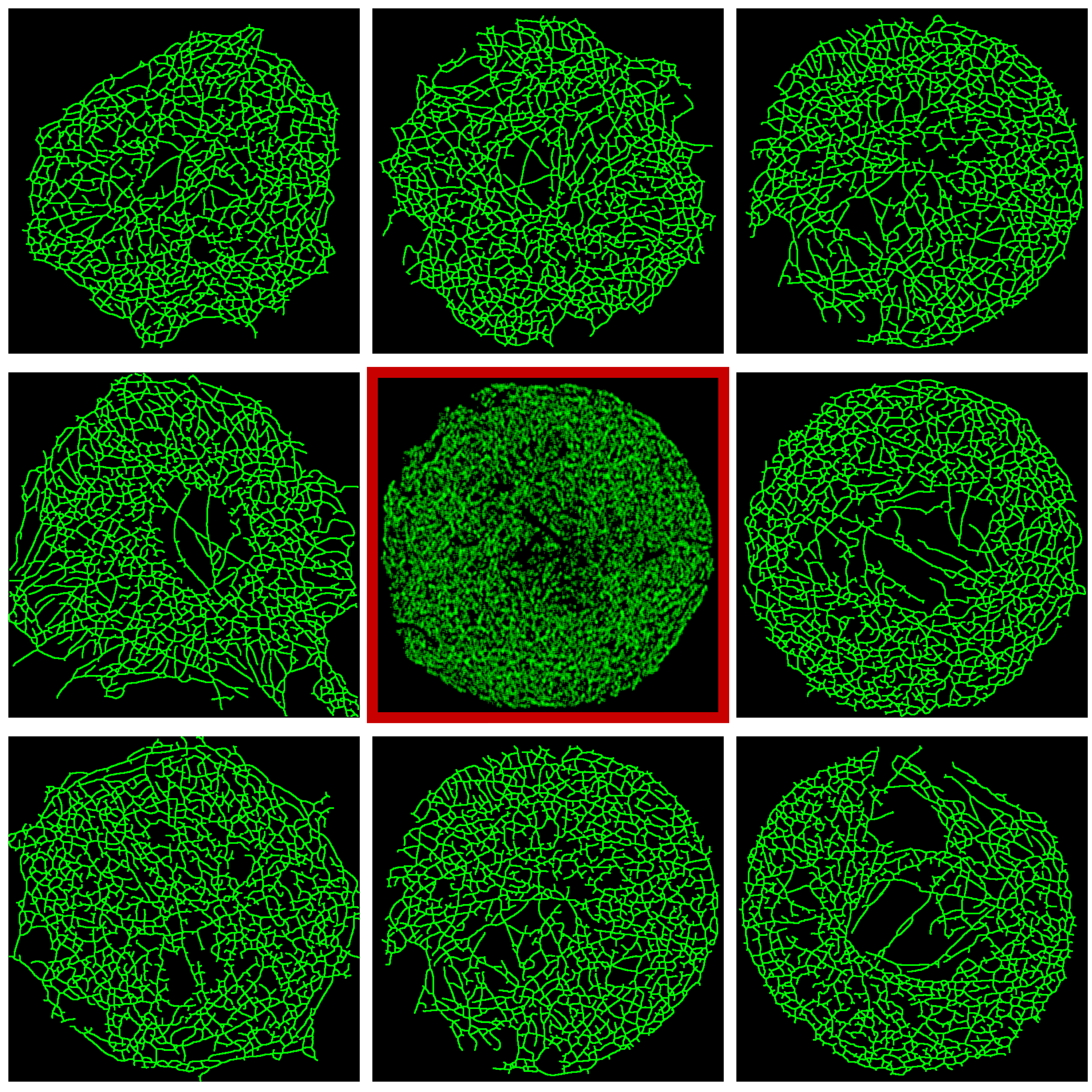}
    \caption{The eight outer confocal microscopy images show fluorescently labelled microtubules in mouse fibroblast cells, which had their graph structure extracted (by use of CytoSeg2.0 \cite{breuer2017system}). The images have a resolution of $642\times 642$. The center image (framed in red) shows their barycenter approximation obtained from the SUA-algorithm with $S=20,000$ samples, which has been mapped back onto a $642\times 642$ grid to produce an image.}
    \label{green_bary}
\end{figure}
\subsubsection{The effect of the number of repeats $R$}\label{subsec:R}
\begin{figure}
    \centering
    \includegraphics[width=\textwidth]{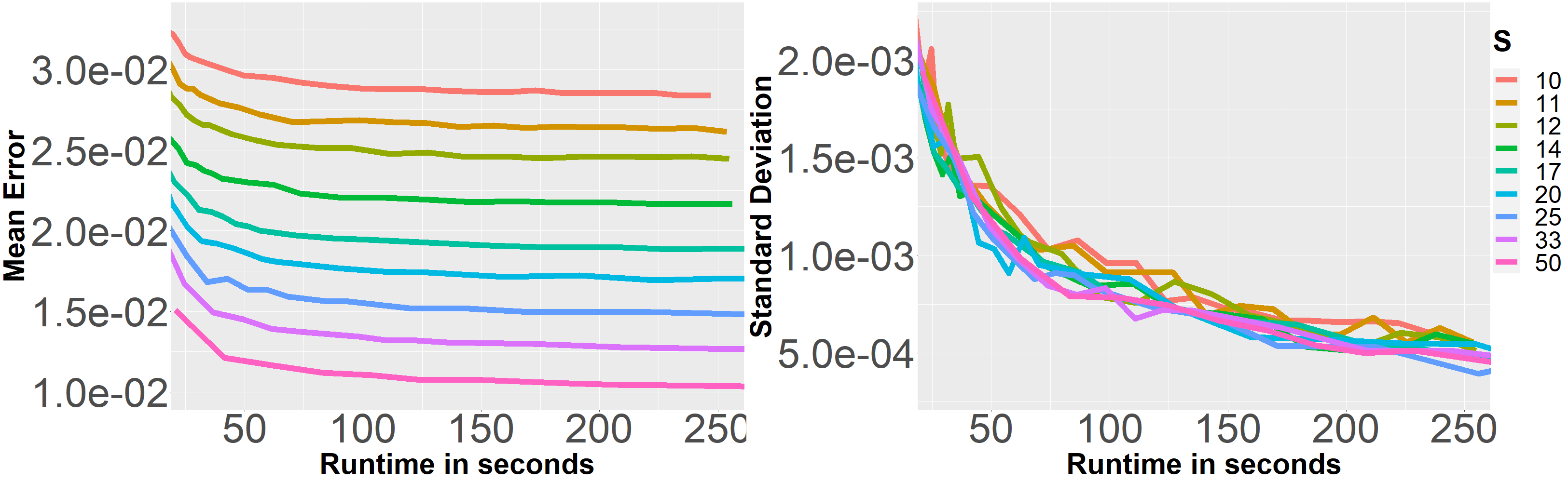}
    \caption{The dependence of the empirical barycenter approximation on the number of repeats $R$. Here, the data measures are $N=20$ multi-nested ellipses supported on $M=100$ support points each. For each pair of $R$ and $S$, $100$ runs have been performed.}
    \label{Rdep_plots}
\end{figure}
Apart from the choice of the sample size $S$, a discussion on the number of repeats $R$ is in order.  Two questions arise in this regard: firstly, how to choose $R$; and secondly, how to combine the $R$ empirical barycenter into a single estimator, when $R>1$.

We address the second matter first.  A natural choice is to take the empirical barycenter that achieves the smallest value of the Fr\'echet functional across the $R$ options.  This approach, however, is computationally demanding and might not even be feasible at all for large problems, since it requires solving optimal transport problems between the barycenter and the data measures.  For this reason, we advocate taking the linear average of the $R$ empirical barycenters, a simple procedure that does not involve optimisation.  In view of \eqref{eq:Fconvex}, it is guaranteed to have objective value that is better than the average value of the $R$ empirical barycenters.  Moreover, our simulations suggest that in fact the linear mean performs better than the best of the $R$ runs for essentially any $R$ and $S$. An example of this can be seen in Figure~\ref{meanvsbestS33}. Intuitively, this superior performance can be attributed to the fact that usually $S<<M$, thus the linear mean has the advantage of having a larger support size, which is closer to  the support size of the original barycenter.

Having discussed the choice of the estimator, we can now proceed to analyse the properties of this estimator with respect to $R$ empirically. To this end, we computed $100$ barycenters for a range of combinations of $R$ and $S$, and estimated the mean and the standard deviation of the estimation error. Naturally, increasing $R$ improves performance, in that both the mean and the variance of the approximation error decrease.  But this observation does not account for the extra computational effort resulting in increasing $R$. For a fixed runtime, increasing $R$ necessarily amounts to decreasing $S$. Figure~\ref{Rdep_plots} shows the empirically observed mean and standard deviation of the error as a function of the runtime, for different choices of $S$. Noticeably, for a fixed runtime the standard deviation of the error is nearly identical for all considered $S$, whilst the mean error clearly decreases. It should be stressed that the fact that the lines at the left plot of Figure~\ref{Rdep_plots} never intersect implies that, in this example, choosing $R=1$ is best for any fixed runtime considered in the simulation.  In other words, any extra available computational effort would be best spent by increasing $S$, and not $R$.  We also point out that if we consider both plots in Figure~\ref{Rdep_plots}, we see that, due to the bias-variance formula, it is also optimal in terms of the mean squared error to increases $S$ instead of $R$. This is, of course, only one example, but we believe that it provides a rather compelling argument for choosing $R$ small. For this reason, we use $R=1$ in the following real data example.

\subsection{A Cell Microscopy Example}
\label{sec::cells}
An important line of research in biophysics concerns the understanding of the structure of filament networks and cytoskeletons of proteins and their behaviour over time when placed under certain conditions (for an overview see e.g., Huber et al.\ \cite{huber2015cytoskeletal}). One can envisage using barycenters here in a number of ways.  One option is to consider the barycenter of a single filament at different time points, thus ``averaging out" variability across time.  Another possibility is to compute the barycenter of a number of different filaments at the same time point, thus averaging over the variability within certain types of proteins.  In that setting one could generate  genetically identical cells from each class of protein.  Although cells at each class are genetic twins, some variability may still arise from mutations.  This variability can be averaged out by taking a barycenter from each protein class, and by doing this at each time point, one can study the evolution of a typical filament (represented by the barycenters) across time.

A modern, high-resolution image usually has hundreds of thousands of pixels, and high quality images can easily extend this well into the millions. In our example, we consider $N=8$ images with resolution of $M=642\times 642> 400,000$, which can be seen in Figure~\ref{green_bary}. Computing the exact barycenter from these $8$ images is intractable.  One could, of course, decrease the resolution in order to obtain a problem of tractable size.  Doing so, however, would result in a substantial information loss, comparable to a decrease in the microscope's resolution, as it would blur the fine structure of the filaments. Instead, we compute an approximate barycenter using our method with $S=20,000$ (Figure~\ref{green_bary}, with red frame in the center).  Upon visual inspection this approximate barycenter seems to represent the 8 filaments well.

\section{Conclusion and Outlook}
Whilst the Wasserstein barycenter shows great potential in geometric data analysis, its tremendous computational cost is a severe hurdle for the OT-barycenter to become a widespread tool in analysing complex data.  Even using modern solvers (potentially with regularisation), one is quickly faced with massive memory and runtime demands which render efficient computations on personal computers impossible for realistic problem sizes.  This is the case for e.g.\ biological imaging, where there is great interest in obtaining sharp high-resolution images of the studied objects, as, for instance, filaments or other protein structures.  Since the number of discretisation points required to achieve a fixed quality grows exponentially with the underlying dimension, the need to alleviate the computation burden is even more paramount in three or higher dimensions in comparison to $\mathbb R^2$. Our SUA-algorithm might be able to bridge the gap between modern applications and the capabilities of state-of-the-art solvers, and may thus open the door to new and interesting data analysis in the natural sciences, where three-dimensional datasets are not uncommon. This is backed by the fact that our convergence rate is independent of the dimension combined with the strong visual performance of our resampled OT-barycenters in three dimensions with relatively low runtime.

Another attractive feature of our method is the possibility to adjust runtime and memory demand depending on the required accuracy of the result, which on the one hand allows to perform quick, rough analysis on personal computers, whilst on the other hand still enables more precise computations on specialised high-performance computing clusters. 

Finally, our simulations suggest even better performance than guaranteed by our theoretical bounds in many settings.  Whilst we have shown that our rates cannot be improved in general, it is still possible, and in fact suggested by simulations, that for certain classes of measures we can obtain better bounds on the expected error if $S$ is a sufficiently small fraction of $M$ (which is the relevant regime in practice). Providing such improved bounds would render our approach even more appealing, and is another clear avenue for further research.

\section*{Acknowledgments}
We wish to thank an associate editor and two reviewers for their constructive feedback.  We are grateful to Sarah K\"oster and Julia B\"orke for providing us with the data used in Section~\ref{sec::cells}. F.\  Heinemann gratefully acknowledges support from the DFG Research Training Group 2088 \textit{Discovering structure in complex data: Statistics meets Optimization and Inverse Problems}.  A.\ Munk gratefully acknowledges support from the DFG CRC 1456 \textit{Mathematics of the Experiment} and the Cluster of Excellence 2067 MBExC \textit{Multiscale bioimaging--from molecular machines to networks of excitable cells}. Y.\ Zemel was supported in part by Swiss National Science Foundation Grant 178220, and in part by a U.K.\ Engineering and Physical Sciences Research Council programme grant.
\begin{appendix}

\section{Auxiliary Proofs}\label{sec:app}
\begin{proof}[Proof of Theorem~\ref{emp_upper}]
The proof is very similar to that of Sommerfeld et al.\ \cite{sommerfeld2019optimal}; we only give the difference here.  The key idea is to improve the bound on the height function $h$ of \cite{sommerfeld2019optimal}.  Using the same notation as \cite{sommerfeld2019optimal}, note that for $1\leq l \leq l_{\max}$, we have
\begin{align*}
\frac{h^p(par(x))-h^p(x)}{\diam(\mathcal X)^p}=q^{-lp} \Bigg( \left( q+ \easysum{j=0}{l_{\max}-l}q^{-j} \right)^p -\left(\easysum{j=0}{l_{\max}-l}q^{-j} \right)^p \Bigg) 
\leq  q^{p-lp} \left( \frac{q}{q-1}\right)^p.
\end{align*}
If $x$ is a leaf, then $h(x)=0$ and the difference is $q^{-pl_{\max}}$ and if $x$ is a root, then $par(x)=x$ and the difference vanishes. Thus, we can start the sum in the definition of $\mathcal{E}$ at $1$ instead of $0$. Also, for $p=1$, we have $q^{l}(h(par(x))-h(x))=q\cdot \diam(\mathcal X)$, where the fractional factor has vanished. Plugging this into the upper bound for the tree metric yields
\begin{align*}
(W_p^{\mathcal{T}}(r,s))^p\leq 2^{p-1}\diam(\mathcal{X})^p \Bigg[ &\easysum{l=1}{l_{\max}} \left( \frac{q}{q-1} \right)^pq^pq^{-lp} \easysum{x \in \tilde{Q}_l}{}\lvert (S_Tr)_x-(S_Ts)_x \rvert \\&+q^{-l_{\max}p} \easysum{x \in \tilde{Q}_{l_{\max}+1}}{}\lvert (S_Tr)_x-(S_Ts)_x \rvert \Bigg].
\end{align*} 
Finally, taking expectations we obtain
\begin{align*}
 \mathbb{E}\left[ W_p^p(\mu,\mu^S) \right]\leq S^{-\frac{1}{2}}2^{p-1}\diam(\mathcal{X})^p q^p \Bigg[ &\easysum{l=1}{l_{\max}} \left( \frac{q}{q-1} \right)^pq^{-lp} \sqrt{\mathcal{N}(\mathcal{X},q^{-l}\diam(\mathcal{X}))} \\&+q^{-(l_{\max}+1)p}M^{\frac{1}{2}} \Bigg].
\end{align*}
For the specific case $p=1$, we have a better bound:
\begin{align*}
 \mathbb{E}\left[ W_1(\mu,\mu^S) \right]\leq S^{-\frac{1}{2}}\diam(\mathcal{X}) q \Bigg[ &\easysum{l=1}{l_{\max}} q^{-l} \sqrt{\mathcal{N}(\mathcal{X},q^{-l}\diam(\mathcal{X}))} +q^{-(l_{\max}+1)}M^{\frac{1}{2}}  \Bigg].
\end{align*}
Imitating the proof of \cite[Theorem 3]{sommerfeld2019optimal} with this improved bound yields the result for $\mathcal X\subset (\mathbb R^D,\|\cdot\|_2)$.  The details are a straightforward adaptation of the work in \cite[Theorem 3]{sommerfeld2019optimal}, and are therefore omitted.
\end{proof}

\begin{proof}[Proof of Lemma~\ref{poly_affine}]
Convexity follows from convexity of $P$, and we also have $|g(t)-g(s)|\le |t-s|\|y\|_1$ by the triangle inequality.  To prove piecewise affineness it suffices to show that there exists $t_0>0$ such that $g$ is affine on $[0,t_0]$.  One can then replace $x$ by $x+t_0y$ and repeat the argument.  Replacing  $y$ by $-y$ yields the result for negative values of $t$.

The minimum in $d_1(x+ty,P)$ is attained by a compactness argument.  Let $z_t\in P$ be such that $g(t)=\|x+ty-z_t\|_1$.

\textbf{Claim:} It is possible to choose $z_t$ and $z_0$ in such a way that $z_t\to z_0$ as $t\to0$.

We will show the stronger statement that as $t\searrow0$
\[
\sup_{z_t\in \argmin d_1(x+ty,P)}\inf_{z_0\in \argmin d_1(x,P)}\|z_t-z_0\|_1\to0. 
\]
By continuity for $t>0$ there exist $z_t\in \argmin \ d_1(x+ty,P)$ and $x_t\in \argmin \ d_1(x,P)$ which attain the supremum and the infimum. Assume that the converse of the claim is true.  Then for some sequence $t_k\to0$, we can choose $z_{t_k}\in \argmin d_1(x+t_ky,P)$ such that for all $k$ and all $z\in\argmin d_1(x,P)$ we have $\|z_{t_k}-z\|_1>\epsilon>0$.  By compactness of $P$ there is a further subsequence $k_l$ such that $z_{t_{k_l}}\to z$.  Since
\[
\|z-x\|_1 = \lim_{l\to\infty} \|z_{t_{k_l}} - x-t_{k_l}y\|_1 = \lim_{l\to\infty} d_1(x+t_{k_l}y,P) = d_1(x,P)
\]
we have $z\in\argmin d_1(x,P)$ and therefore
\[
0<\epsilon\le \liminf_{l\to\infty} \|z_{t_{k_l}} - z\|_1 =0,
\]
a contradiction.  This proves the claim.

Since $P$ is a polyhedron, for all $z_0\in P$ there exists $r_0(z_0)>0$ such that for all $v\in \R^L$ with $\|v\|_1\le r_0$,
\[
z_0+v \in P \quad \Longrightarrow z_0+2v\in P.
\]
By the previous claim, there exist $z_0\in \argmin d_1(x,P)$, $z_t\in \argmin d_1(x+ty,P)$, and $T_0>0$ such that for all $t\in[0,T_0]$,
\[
\|x + ty - z_t\|_1 = d_1(x+tv,P),\quad \|x-z_0\|_1 = d_1(x,P),\qquad \|z_t-z_0\|_1\le r_0=r_0(z_0).
\]
Therefore $z_0+2(z_t-z_0)\in P$ for all $t\in[0,T_0]$ and
\[
d_1(x,P) + d_1(x+2ty,P)
\le \|x-z_0\|_1 + \|x+2ty-z_0-2(z_t-z_0)\|_1= \|x^0\|_1 + \|x^0 + 2u^t\|_1,
\]
where $x^0=x-z_0$ and $u^t=ty - (z_t-z_0)\to0$.  Define
\[
I_+=\{1\le i\le L:x^0_i>0\},\quad
I_-=\{1\le i\le L:x^0_i<0\},\quad 
I_0=\{1\le i\le L:x^0_i=0\}.
\]
For $t$ sufficiently small (and such that $t\le T_0$) we have $2\|u^t\|_1\le \underset{i\notin I_0}{\min}|x_i^0|$ (where we define the minimum to be infinite if  $x= z_0$) and therefore
\begin{align*}
d_1(x,P) + d_1(x+2ty,P)
&\le \sum_{i\in I_+} (x^0_i + x^0_i + 2u^t_i)
+ \sum_{i\in I_-} (-x^0_i - x^0_i - 2u^t_i)
+ \sum_{i\in I_0} 2|u^t_i|\\
&= 2\sum_{i\in I_+} (x^0_i + u^t_i)
+ 2\sum_{i\in I_-} (-x^0_i -u^t_i)
+ 2\sum_{i\in I_0} |u^t_i|
\\&=2d_1(x+ty,P).
\end{align*}
Thus, there exists $0<t_0\le T_0$ such that for all $t\in [0,t_0]$,
\[
g(0) + g(2t)
\le 2g(t).
\]
Since $g$ is convex and finite on $\R$, this implies that $g$ is affine on $[0,t_0]$.  As stated at the beginning of the proof, this shows that $g$ is piecewise affine on the real line.
\end{proof}

\begin{proof}[Proof of Lemma~\ref{slope_bound}]
If $\mu$ is optimal, then we can choose $\mu^*=\mu$ and there is nothing to prove.  Hence we may assume that $\mu$ is not optimal.  Let $P\subset \mathbb{R}^L$ be the feasible polytope corresponding to \eqref{eq:linProgBary}. 
To each $\mu \in \mathcal{P}(\mathcal{\mathcal{C}})$ corresponds a $\pi \in P$ such that $c^T\pi=F(\mu)$. 
Fix an element $\pi^* \in \argmin d_1(\pi,\mathcal{M})$, from which we can construct a minimiser $\mu^*$ of $F$. It holds that
\[
F^p(\mu)-F^p(\mu^*)
=c^T \pi - c^T \pi^* =\|\pi-\pi^*\|_1 \frac{c^T \pi - f^*}{\|\pi-\pi^*\|_1}\ge \|\pi-\pi^*\|_1\psi(\pi),
\]
where
\[
\psi (\pi)= \frac{c^T\pi-f^*}{d_1(\pi,\mathcal M)},
\qquad \pi\in P\setminus \mathcal M.
\]
Next, we aim to show that the infimum of $\psi$ is attained at a vertex of $P$. 
If $\pi$ is not a vertex of $P$, then there exists a vector $v\in \mathbb{R}^L\backslash \{0\}$ such that $\pi +tv \in P$ for all $t \in [-1,1]$. Let 
\begin{align*}
\phi(t):=\psi(\pi+tv)=
\frac{c^T \pi-f^*+tc^T v}{d_1(\pi+tv,\mathcal M)}:=\frac{at+b}{g(t)}
,\qquad b=c^T\pi - f^*,
\end{align*}
where the nominator is affine by definition and the denominator, $g(t)$, is piecewise affine by Lemma~\ref{poly_affine}. Since $g$ is continuous and convex, we have $g(t)=\alpha_+t+\beta$ for small $t\ge0$ and $g(t)=\alpha_-t+\beta$ for small $t\le 0$, with $\alpha_-\le\alpha_+$ and $\beta>0$.  Taking the derivative we obtain
\begin{align*}
\phi^\prime_\pm (t)=\frac{a(\alpha_\pm t+\beta)-\alpha_\pm (at+b)}{(\alpha_\pm t+\beta )^2}=\frac{a\beta-b\alpha_\pm}{(\alpha_\pm t+\beta)^2},
\end{align*}
where $\phi'_+$ and $\phi'_-$ denote the right and left derivatives. We now distinguish three cases. If $\phi^\prime_+ (0)<0$, then $\psi(\pi+tv)<\psi(\pi)$ for $t>0$ small enough. If $\phi^\prime_- (0)>0$, then $\psi(\pi+tv)>\psi(\pi)$ for $t<0$ small enough. In both cases $\pi$ is not a minimiser.  Since $\phi'_+(0)\le \phi'_-(0)$, the only other possibility is that both derivatives vanish. Replacing $v$ by $-v$ if necessary, we may assume that $a\ge0$. We also have $b>0$ and $\beta >0$, since $\pi  \in P\setminus\mathcal{M}$ by assumptions. Thus it holds $\alpha_\pm=\alpha=a\beta/b\ge0$.

We can now move in direction $v$ (i.e., increasing $t$) until one of two things happens. Either we can no longer move in direction $v$ without leaving $P$, or the denominator of $\phi$ has changed since we moved into a different affine segment of $g$ ($t$ cannot go to infinity since $P\subseteq[0,1]^L$ is bounded by definition, and we cannot reach $\mathcal M$ because $a\ge0$.). In the first case, we have reached a face of $P$ (if $\pi$ is in the interior of $P$) or a strict subface of one of the faces of $P$ containing $\pi$ (if $\pi$ is on the boundary of $P$) and can now search a new direction $v$ in this (sub)face. Since $P$ is a polyhedron, this can only happen finitely many times until we reach a vertex of $P$.  In the second case we reach a point $t_0>0$ at which the right derivative of $g$ is strictly larger than $\alpha$ (by convexity of $g$).  Thus $g(t_0)\ge g(0)>0$, the right derivative of $\phi$ becomes negative, and for $\epsilon$ small
\[
\psi(\pi+(\epsilon + t_0)v)
=\phi(\epsilon+t_0)
<\phi(t_0)
=\phi(0)
=\psi(\pi),
\]
and consequently $\pi$ is not a minimiser.

All in all, for any $\pi\in P\setminus \mathcal M$, either $\pi$ does not minimise $\psi$, or there exists a vertex $v$ of $P$ that is not in $\mathcal M$ and such that $\psi(v)\le \psi(\pi)$.  Hence
\begin{align*}
\underset{\pi \in P\backslash \mathcal{M}}{\inf} \psi(\pi)
= \underset{\pi \in V\backslash V^*}{\min}\ \psi (\pi) =\underset{v \in V\backslash V^*}{\min} \ \frac{c^Tv-f^*}{d_1(v,M_f)}=:\widetilde{C}_P,
\end{align*}
and therefore
\begin{align*}
F^p(\mu)-F^p(\mu^*)
=c^T (\pi - \pi^*) 
\ge  \psi(\pi) \lVert \pi- \pi^* \rVert_1
\ge  \widetilde{C}_P \lVert \pi- \pi^* \rVert_1.
\end{align*}
To relate $\|\pi-\pi^*\|_1$ to the Wasserstein distance recall that $\pi$ encodes a measure $\mu$ as well as the optimal transport plan $T^i$ between $\mu$ and each measure $\mu_i$.  Hence
\begin{align*}
\lVert \pi- \pi^* \rVert_1&\geq \easysum{i=1}{N} \lVert T^{(i,1)}-T^{(i,2)} \rVert_1 + \easysum{x \in \mathcal{C}}{} \lvert \mu(x)- \mu^*(x) \rvert \\ &\geq N  \easysum{x \in \mathcal{C}}{} \lvert \mu(x)- \mu^*(x) \rvert + \easysum{x \in \mathcal{C}}{} \lvert \mu(x)- \mu^*(x) \rvert = (N+1)\easysum{x \in \mathcal{C}}{} \lvert \mu(x)- \mu^*(x) \rvert,
\end{align*}
where $T^{(i,1)}$ is an optimal transport plan between $\mu_i$ and $\mu$ and $T^{(i,2)}$ is an optimal plan between $\mu_i$ and $\mu^*$. Thus,
\begin{align*}
F^p(\mu)-F^p(\mu^*)&\geq \widetilde{C}_P  (N+1)\easysum{x \in \mathcal{C}}{} \lvert \mu(x)- \mu^*(x) \rvert \\ &= 2 (N+1) \widetilde C_P \ \mathrm{TV}(\mu,\mu^*) \ge   2 C_P W_p^p(\mu,\mu^*),
\end{align*}
with $C_P=(N+1)\widetilde C_P\diam(\mathcal X)^{-p}$.  Finally, $\widetilde C_P>0$ because it is a minimum of finitely many positive numbers.  Positivity of $C_P$ follows.
\end{proof}

\end{appendix}

{\footnotesize  
\bibliographystyle{plain}
\bibliography{emp_bary_arxiv_revised}
\end{document}